\documentclass[11pt]{article}
\usepackage[margin=1in]{geometry}
\usepackage[utf8]{inputenc}
\pdfoutput=1

\usepackage{amsmath,amsthm,enumerate,graphicx,amsfonts}
\usepackage{algorithm}
\usepackage[noend]{algorithmic}
\usepackage{enumitem}
\usepackage{algorithm}
\usepackage{algorithmic}
\usepackage{times}              % Font family selection
\usepackage{latexsym}
\usepackage{amsmath}
\usepackage{amsfonts}
\usepackage{amsthm}
\usepackage{epsfig}
\usepackage{graphicx}
\usepackage{cite}
\usepackage{caption}
\usepackage{titling}
\usepackage{fancyhdr}
\usepackage{enumerate}

\newtheorem{theorem}{Theorem}[section]
\newtheorem{proposition}[theorem]{Proposition}
\newtheorem{lemma}[theorem]{Lemma}

\newtheorem{corollary}[theorem]{Corollary}
\newtheorem{definition}[theorem]{Definition}
\newtheorem{observation}[theorem]{Observation}

\DeclareMathOperator*{\E}{\mathbb{E}}

\DeclareMathOperator*{\poly}{poly}
\DeclareMathOperator*{\polylog}{polylog}

\newcommand{\A}{\mathcal{A}}
\newcommand{\F}{\mathcal{F}}
\newcommand{\C}{\mathcal{C}}

\renewcommand{\S}{\mathcal{S}}

\newcommand{\M}{\mathcal{M}}

\newcommand{\bE}{\ensuremath{\mathbb{E}}}
\newcommand{\cost}{\text{cost}}
\newcommand{\OPT}{\text{OPT}}

\begin{document}
\title{Dependent randomized rounding for clustering and partition systems with knapsack constraints\thanks{
Research supported in part by NSF Awards CNS-1010789, CCF-1422569, CCF-1749864, and CCF-1918749, and by research awards from Adobe, Inc., Amazon, and Google}}

\author{  
David G. Harris\thanks{Department of Computer Science, University of Maryland,
College Park, MD 20742. Email:
\texttt{davidgharris29@gmail.com}}
\and
Thomas Pensyl\thanks{Bandwidth, Inc. Raleigh, NC. Email: \texttt{tpensyl@bandwidth.com}}
\and
Aravind Srinivasan\thanks{Department of Computer Science and Institute for Advanced Computer Studies, University of Maryland, College Park, MD 20742.  Email: 
\texttt{asriniv1@cs.umd.edu}} 
\and Khoa Trinh\thanks{Google, Mountain View, CA 94043. Email: \texttt{khoatrinh@google.com}}}

\date{}
\maketitle

\vspace{-0.4in}

\begin{abstract}
Clustering problems are fundamental to unsupervised learning. There is an increased emphasis on \emph{fairness} in machine learning and AI; one representative notion of fairness is that no single group should be over-represented among the cluster-centers. This, and much more general clustering problems, can be formulated with ``knapsack" and ``partition" constraints. We develop new randomized algorithms targeting such problems, and study two in particular: multi-knapsack median and multi-knapsack center. Our rounding algorithms give new approximation and pseudo-approximation algorithms for these problems. 

One key technical tool, which may be of independent interest, is a new tail bound analogous to Feige (2006) for sums of random variables with unbounded variances. Such bounds can be useful in inferring properties of large networks using few samples. 
\end{abstract}

% \newpage

%\tableofcontents
%\newpage

This is an extended version of a paper which appeared in the Proc. 23rd International Conference on
Artificial Intelligence and Statistics (AISTATS 2020)

\section{Introduction}
Clustering is a fundamental technique in unsupervised learning. Our goal is to systematically study clustering with knapsack constraints, particularly in light of fairness. Consider a data-clustering problem, with a set $\F$ of potential cluster-centers with $m$ non-negative cost functions $M_1, \dots, M_m$, a data-set $\C$, and a symmetric distance-metric $d$ on $\F \cup \C$. Our goal is to choose a set $\S \subseteq \F$ of cluster-centers to minimize the distances   $d(j, \S)$, for $j \in \mathcal C$, while satisfying the budget constraint $\sum_{i \in \S} M_{k}(i) \leq 1$ for each $k$.

 The metric used to boil down the values $d(j, \S)$ into an objective function $\cost(\S)$ is problem-specific. We refer to the $m \times | \mathcal F |$ matrix $M$ as the \emph{knapsack-constraint matrix}; we refer to the case $m > 1$ as  \emph{multi-knapsack} and the case $m = 1$ as \emph{single-knapsack}.  \emph{We typically view $m$ as ``small", e.g., as a constant.} The RHS value $1$ for the budget constraint is just a normalization.    Note that it is possible to have $\F = \C$. 

 As a representative example, which is in fact one of the main motivations behind this work, consider the following scenario. As in classical clustering, we want to select a small number of clusters of small radius for a population $\mathcal C$. In addition, we are given a collection of $m$ groups $A_1, \dots, A_m$, such that we do not wish to disproportionately select our centers $\mathcal S$ from any  group. For instance, in healthcare facility location in the face of an epidemic, we may not want to open too many facilities in a geographic region, or with certain types of equipment, or near only some groups of patients, etc.  Thus, for each $k = 1, \dots, m$, we have a constraint $|\S \cap A_k| \leq t_k$ where $t_k$ is a target value which is not much larger than the ``fair'' proportion of $A_k$ compared to the general population.
 These can be represented as knapsack constraints; due to the normalization we use, the cost functions are given by $M_{k}(i) = 1/t_k$ if $i \in A_k$ and $M_{k}(i) = 0$ otherwise. 
  
A common strategy for clustering problems is to first solve a related linear program (LP), leading to a fractional configuration $y \in [0,1]^{\mathcal F}$ over $\mathcal F$ satisfying the knapsack constraints. Here, the LP suggests the fractional extent $y_i$ for $i \in \mathcal F$ to be chosen as a center. Based on this solution $y$, we partition $\mathcal F$ into groups, which represent sets of cluster-centers ``close to'' certain items $j \in \C$. Finally, we use some randomized-rounding algorithm to convert the fractional solution $y$ to an integral solution $Y \in \{0,1 \}^{\mathcal F}$ representing the chosen solution $\S$, while ensuring that each group gets a selected center. (This grouping is needed because if some client $j$ has no nearby opened facilities, then $d(j, \S)$ may be very large.) 

The vector $Y$ should also have other probabilistic properties related to the vector $y$, for example satisfying $\E[Y] = y$ coordinatewise. Clustering problems pose a particular challenge for randomized rounding because they are fundamentally \emph{non-linear}; the distance from $j$ to its closest center depends on the \emph{joint} behavior of the selected centers. Consequently, the rounding process should also ensure strong independence properties, beyond just bounds on the expected values of individual coordinates of $Y$.

One obvious LP rounding method for the LP is for each group to independently choose a single center $i$, with probability proportional to $y_i$. This has optimal independence properties, but completely ignores the knapsack constraints.  This paper focuses on a new randomized rounding strategy for the LP. We develop two key probabilistic techniques, both of which are quite general and may be of independent interest:  (i) a new  ``Samuels-Feige'' type of concentration inequality for \emph{unbounded} random variables; and (ii) a new rounding algorithm in the presence of knapsack constraints plus a single partition constraint.  
 
As we will see, we cannot simultaneously achieve the goals of exactly preserving the knapsack constraints and mimicking the probabilistic guarantees of independent selection. Nevertheless, we obtain significantly stronger guarantees compared to previous algorithms. We summarize these next, and then discuss our new clustering results.

 \subsection{Dependent rounding and independence}
The general problem of randomized rounding while preserving hard combinatorial constraints often goes by the name ``dependent rounding''; see, e.g., \cite{byrka2015,DBLP:conf/ipco/ByrkaSS10,charikar_dependent}.   In the most straightforward form of dependent rounding, which we call \emph{cardinality rounding}, we have a fractional solution $x \in [0,1]^n$ that we wish to round to an integral vector $X \in \{0, 1 \}^n$ such that $\E[X] = x$ and the cardinality constraint $\sum_i X_i = \sum_i x_i$ holds with probability one. For example, Charikar and Li \cite{charikar_dependent} applied cardinality rounding as part of their $3.25$-approximation algorithm for $k$-median. 

The cardinality constraint can easily be replaced by a single knapsack constraint  \cite{AS,srin:level-sets}.  Over the last two decades, increasingly-sophisticated dependent-rounding techniques have been used for optimization problems over various types of polytopes; see, e.g., \cite{AS,srin:level-sets, gandhi:depround, DBLP:journals/siamcomp/CalinescuCPV11, DBLP:conf/soda/ChekuriVZ11, Bansal18}. 

Our new rounding algorithm, which we call \emph{Knapsack-Partition Rounding (KPR)}, generalizes this setting in two distinct ways: it allows multiple knapsack constraints, and it allows a partition matroid constraint. Formally, we define a \emph{knapsack-partition system} to be a partition $\mathcal G$ over a ground-set $U$, along with a real-valued $m \times |U|$ matrix $M$ and a fractional vector $y \in [0,1 ]^U$ satisfying $y(G) = 1$ for each $G \in \mathcal G$.

Our overarching question is: \emph{how well can we approximate independence while preserving the knapsack and partition constraints?}    In an ideal scenario, we would like to generate a random vector  $Y \in \{0,1 \}^U$ satisfying the following desiderata (which are, to a certain extent, mutually incompatible and unattainable -- indeed, our algorithm instead guarantees certain related conditions (E1)--(E6)):
\begin{itemize}
\item[(D1)] $\E[Y_j] = y_j$ for every $j \in U$
\item[(D2)] The random variables $Y_j$ are negatively correlated, in some sense.
\item[(D3)] $Y(G) = 1$ for $G \in \mathcal G$.
\item[(D4)] $M Y = M y$
\end{itemize}

To explain these further,  consider the \emph{independent-selection} rounding strategy. Formally, we define $Y = \textsc{IndSelect}(\mathcal G,y)$ to be the vector obtained by selecting, independently for each block $G$, exactly one item $j$ from $G$ and setting $Y_j = 1$, so that each item $j$ is selected with probability $y_j$. All other (non-selected) items have $Y_j = 0$. This is a valid probability distribution as $y(G) = 1$ and the entries of $y$ are in the range $[0,1]$.  The vector  $Y$ satisfies desiderata (D1), (D2), (D3) perfectly. However, it only weakly satisfies (D4): specifically, the value of $M_k Y$ will be a sum of negatively-correlated random variables, which can deviate significantly from its mean $M_k y$.

As we will later discuss in Section~\ref{dep-round-sec}, the cardinality-rounding setting can be viewed as a special case of knapsack-partition constraints. The standard cardinality-rounding algorithm perfectly satisfies (D1) and (D4); (D3) is not meaningful in this case. For (D2), it satisfies a limited but important form of negative correlation known as  the \emph{negative cylinder property} \cite{DBLP:conf/focs/ChekuriVZ10, byrka2015, srin:level-sets, gandhi:depround}; namely, for any set $S \subseteq U$ the rounded variables $X_i$ satisfy the conditions
\begin{align}\label{eq:thm-neg_correlation}
\E \Bigl[\prod_{i\in S}X_i \Bigr] \leq \prod_{i\in S}x_i, ~~~\text{and}~~~ \E \Bigl[\prod_{i\in S}(1-X_i) \Bigr] \leq \prod_{i\in S}(1-x_i).
\end{align}

Our clustering results will require more general forms of negative correlation.  Ideally, we would like for arbitrary disjoint sets $S, T \subseteq U$ to satisfy a similar ``near-independence'' property:
\begin{equation}
\label{eqn:bsr-small-error0}
\E \Bigl[\prod_{i\in S}X_i \prod_{i \in T} (1 - X_i) \Bigr] \approx \prod_{i\in S} x_i \prod_{i \in T} (1 - x_i)
\end{equation}

For example, Byrka et.\ al.\ \cite{byrka2015} showed this property for the cardinality-rounding algorithm with a random permutation of the input vector, in some parameter ranges.  Note that Eq.~(\ref{eqn:bsr-small-error0}) cannot be preserved exactly in an integral solution; for example, if $x_1 = x_2 = 1/2$, then any integral solution must either satisfy $\E[X_1 (1 - X_2)] \geq 1/2$ or $\E[(1-X_1) X_2] \geq 1/2$. This is part of the reason why property (D2) and general negative correlation are much more challenging than the negative cylinder property. 

To overcome this fundamental barrier, our KPR algorithm terminates with a vector which has a small number $t$ of fractional entries.   To provide intuition, let us discuss how our rounding algorithm works in  the cardinality-rounding setting. (We emphasize that it can handle much more general scenarios.)  The precise sense in which it  satisfies Eq.~(\ref{eqn:bsr-small-error0}) is somewhat technical, but, as one example, we get
\begin{equation}
\label{eqn:bsr-small-error}
\E[ \prod_{i \in S} X_i \prod_{i \in T} (1 - X_i) ] \leq \Bigl( \prod_{i \in S} x_i \prod_{i \in T} (1-x_i) \Bigr) + O(1/t)
\end{equation}
We emphasize that Eq.~(\ref{eqn:bsr-small-error}) is only a simplified form of our results. In particular, we will  handle cases where $S \cup T$ has a large size but only a few elements are ``significant.'' 

Let us briefly compare KPR with other dependent rounding algorithms. The first main genre of such algorithms is based on variants of the Lov\'{a}sz Local Lemma (e.g.,  \cite{DBLP:conf/focs/HarrisS13,DBLP:journals/siamcomp/LeightonLRS01,DBLP:journals/siamcomp/Srinivasan06}). These algorithms have very good independence properties, but also lead to knapsack violations on the order of  the ``standard deviation.''    A second genre of algorithm is based, like KPR, on Brownian motion in the constraint polytope. These algorithms are often targeted to discrepancy minimization, see e.g., \cite{beck-fiala,klrtvv, DBLP:conf/ipco/BansalN16,Bansal18}, where the central  goal is to show concentration bounds on linear functions of the variables. This is closely related to pairwise correlation (covariance) among the variables. Our algorithm gets tighter bounds and finer negative correlation properties, including correlations among many variables, by taking advantage of the special properties of the knapsack-partition constraints.
 
\subsection{Additive pseudo-approximation}
Knapsack constraints can be intractable to satisfy exactly,  and so \emph{pseudo-approximations} (i.e., solutions which only approximately satisfy the knapsack constraints) are often used instead.  Many previous algorithms (e.g., \cite{byrka2015})  have focused on what we refer to as \emph{$\epsilon$-multiplicative pseudo-solutions}: namely that $\sum_{i \in \S} M_{k}(i) \leq 1 + \epsilon$ for each $k$. This should be distinguished from a true approximation algorithm, which finds a feasible solution whose objective function is within some constant factor of the optimal one. 

As we have discussed, our rounding algorithm  does not generate a fully-integral  vector $Y$,  it only produces a vector $\tilde Y$ which is ``mostly'' integral, that is, $\tilde Y \in [0,1]^U$ has only a small (essentially constant) number of fractional entries.  This is critical to overcoming the tradeoff between probabilistic independence and satisfying the knapsack constraints. This naturally leads to an alternative, \emph{additive} form of knapsack pseudo-approximation. We define this formally as follows:
\begin{definition}[$q$-additive pseudo-approximation]
For a single-knapsack constraint (vector of weights $w$ with capacity $1$), a set $\S$ is a $q$-additive pseudo-solution if it has the form $\S = \S_0 \cup \S_1$, where $\sum_{i \in \S_0} w_i \leq 1$ and $|\S_1| \leq q$. (Equivalently, $\S$ satisfies the budget constraint after removing its $q$ highest-weight items.)

For a multi-knapsack constraint $M = M_1, \dots, M_m$, we say that $\S$ is an \emph{$q$-additive pseudo-solution for $M$} if it is a $q$-additive solution for each of the $m$ knapsack constraints $M_1, \dots, M_m$ separately. That is, for each $k = 1, \dots, m$ we have $\S = \S_0^{(k)} \cup \S_1^{(k)}$  where  $\sum_{i \in \S_0^{(k)}} M_{k}(i) \leq 1$ and $|\S_1^{(k)}| \leq q$.
\end{definition}

Such additive pseudo-solutions are naturally connected to the \emph{method of alteration} in the probabilistic method, where we delete/alter some items in a random structure to establish a desired property \cite{alon-spencer:tpm}.  Additive pseudo-approximations have appeared implicitly in prior algorithms, e.g., \cite{DBLP:conf/stoc/LiS13, krishnaswamy_km}. We can summarize some of their advantages here, from both practical and technical points of view. 
 
First, additive pseudo-approximation can be a useful tool to obtain true approximations. Gven a $q$-additive pseudo-solution, we can often ``fix'' the $q$ additional items in some problem-specific way. If $q$ is small, this may incur only a small overhead in the cost or computational complexity. This strategy was used in the $k$-median approximation algorithm of \cite{DBLP:conf/stoc/LiS13}. We use it here for our (true) approximation algorithm for knapsack center. These problems are described in 
Section~\ref{sec:clustering-defns}. 

As another example, there is a common strategy to obtain a multiplicative pseudo-approximation by ``guessing'' -- exhaustively enumerating -- the  ``big'' items in an optimal solution, i.e., items with $M(i) > \rho$ for some parameter $\rho$.  Then, a $q$-additive pseudo-approximation for the residual problem yields a $\rho q$-multiplicative pseudo-approximation to the original problem. This approach can be more efficient than generating the multiplicative pseudo-solution  directly. The reverse direction does not hold, in general; there does not seem to be any way to go from multiplicative to additive pseudo-solutions.

To see a practical advantage of additive pseudo-approximation, consider our motivating example concerning fair representation. In this setting, a $q$-additive pseudo-solution $\mathcal S$ leads to a relatively modest violation of the fairness constraint, namely, it has $|\S \cap A_k| \leq t_k + q$ for each population $A_k$ and associated target value $t_k$. By contrast, an $\epsilon$-multiplicative pseudo-solution $\S$ would give a substantially larger violation, namely $|\S \cap A_k| \leq t_k (1 + \epsilon)$.

\subsection{Clustering results} 
\label{sec:clustering-defns}
In describing our clustering results, it is convenient to use the language of classical facility location. We refer to $\C$ as a set of \emph{clients}, $\F$ as a set of \emph{facilities}, and we say that $i \in \F$ is \emph{open} if $i$ is placed into the solution set $\S$. The distance $d(j, \S)$ for a client $j$ can be interpreted as the connection cost of $j$ to its nearest open facility. 

We study two clustering problems in particular: the \emph{knapsack median} and \emph{knapsack center} problems. In the knapsack median problem, we minimize the \emph{total} connection $\cost(\S) = \sum_j d(j, \S)$ subject to $m$ knapsack constraints.   The knapsack center problem is the same except that the objective is to minimize $\cost(\S) = \max_{j \in \C} d(j, \S)$ instead of the sum $\sum_j d(j, \S)$.

Knapsack median was first studied by Krishnaswamy et.\ al.\ \cite{krishnaswamy_km}, who obtained an additive pseudo-approximation with an approximation factor of $16$. The current best true approximation factor is $7.08$ for the single-knapsack case,  due to \cite{KLS}. The special case when all facilities in $\F$ have unit weight and $m = 1$, known as the classical $k$-median problem, can be approximated to within a factor of $2.675+\epsilon$ \cite{byrka2015}.   Our KPR rounding algorithm gives the following pseudo-approximation results:
\begin{theorem}
\label{thm:knapsack-median-one-budget}
Let $\gamma, \epsilon \in (0,1)$.  For single-knapsack median, there is a polynomial-time algorithm to obtain an $O(1/\gamma)$-additive pseudo-solution $\S$ with $\cost(\S) \leq (1+\sqrt{3}+\gamma) \cdot \OPT \leq 2.733 \cdot \OPT$ and an algorithm with $n^{O(\epsilon^{-1} \gamma^{-1})}$ runtime to obtain an $\epsilon$-multiplicative pseudo-solution $\S$ with  $\cost(\S) \leq (1 + \sqrt{3} + \gamma) \cdot \OPT$.
\end{theorem}

We also consider multi-knapsack median, where it is NP-hard to obtain a true approximation. We apply KPR for a key step in an algorithm of Charikar \& Li \cite{charikar_dependent} to get an additive pseudo-approximation. This can also be leveraged into a multiplicative pseudo-approximation. We summarize these results as follows:
\begin{theorem}
  \label{thm:multi-knapsack-medianx}
Let $\gamma, \epsilon \in (0,1)$.  For multi-knapsack median, there is a polynomial-time algorithm to obtain an $O(\tfrac{m}{\sqrt{\gamma}})$-additive pseudo-solution $\S$ with $\cost(\S) \leq (3.25 + \gamma) \cdot \OPT$, and an algorithm with $n^{O(\frac{m^2}{\epsilon \sqrt{\gamma}})}$ runtime to obtain an $\epsilon$-multiplicative pseudo-solution $\S$ with $\cost(\S) \leq (3.25 + \gamma) \cdot \OPT$.
\end{theorem}

By contrast, independent selection in the Charikar-Li algorithm would take $n^{\tilde O(m/\epsilon^{2})}$ runtime for an $\epsilon$-multiplicative approximation (see Theorem~\ref{jkj1}). This illustrates how additive pseudo-approximation can be useful for efficient multiplicative pseudo-approximations; in particular, this gives a better dependence on the parameter $\epsilon$ (though a worse dependence on parameters $m$ and $\gamma$).

The single-knapsack center problem was first studied by Hochbaum \& Shmoys in  \cite{DBLP:journals/jacm/HochbaumS86}, under the
name ``weighted $k$-center''. They gave a $3$-approximation algorithm and proved that this is best possible unless $\text{P} = \text{NP}$; see also \cite{DBLP:journals/tcs/KhullerPS00}. Our approximation algorithms ensure each client has better bounds on \emph{expected} connection distance, in addition to the usual bound on  maximum distance. We summarize this as follows:

\begin{theorem}
  \label{res:standardKC}
  Let $\gamma \in (0,1)$. There is an algorithm for single-knapsack center with $n^{\tilde O(1/\gamma)}$ runtime which returns a feasible solution $\S$  such that every client $j \in \C$ has
  $$
  \E[d(j, \S)] \leq (1 + 2/e + \gamma) \cdot \OPT, \qquad \qquad d(j, \S) \leq 3 \cdot \OPT \text{ with probability one. }
  $$
\end{theorem}

More recently, Chen et.~al.~\cite{jianli_kc} considered the multi-knapsack center problem. They showed that it is intractable to obtain a true constant-factor approximation, and gave a multiplicative pseudo-approximation with approximation ratio $3$. We obtain a number of new pseudo-approximations for this setting.
\begin{theorem}
  \label{res:MKC}
Let $\gamma, \epsilon \in (0,1)$. For multi-knapsack median, we describe three algorithms to generate different types of pseudo-solutions $\S$ such that every client $j \in \C$ has  $\E[d(j, \S)] \leq (1 + 2/e + \gamma) \cdot \OPT$ and $d(j, \S) \leq 3 \cdot \OPT$ with probability one.
\begin{enumerate}
\item[(a)] A polynomial-time algorithm for an $\tilde O(m/\sqrt{\gamma})$-additive pseudo-solution.

\item[(b)]  An algorithm with $n^{\tilde O(m^2 / \gamma)}$ runtime for an $\tilde O(\sqrt{m})$-additive pseudo-solution.

\item[(c)]   An algorithm with $n^{\tilde O(m^{3/2}/\epsilon + m^2/\gamma)}$ runtime for an $\epsilon$-multiplicative pseudo-solution.
  \end{enumerate}
\end{theorem}

Our algorithms thus give finer guarantees:  all clients have distance $3 \cdot \OPT$ to an open facility with probability one and also all clients have expected connection cost at most $(1+2/e) \cdot \OPT \approx 1.74 \cdot \OPT$.  
This can be helpful in \emph{flexible} facility location, such as a streaming-service provider periodically reshuffling its service locations. It can also be interpreted as a type of fairness in clustering, where the fairness is in terms of individual users instead of demographic groups. Note that the constant factor $1 + 2/e$ cannot be improved unless $\text{P} = \text{NP}$, even in the $k$-supplier setting \cite{harris-2}. 

\subsection{Notation}

In the context of clustering problems, we let $n = |\mathcal F \cup \mathcal C|$. For a client $j \in \C$ and a real number $x \geq 0$, we define the \emph{facility-ball} $\mathcal B(j, x) = \{ i \in \F \mid d(i, j) \leq x \}$.  For a metric $d$ and a set $Y$, we write $d(x,Y) = \min_{y \in Y} d(x,y)$. 

For a non-negative integer $t$, we write $[t] = \{1, \dots, t \}$. For a set $X$ and an $t$-dimensional vector $y$, we write $y(X) = \sum_{i \in X} y_i$.  The \emph{support} of $y$ is defined to be the set of indices $i$ where $y_i \neq 0$.

Given an $m \times n$ budget matrix $M$, we write $M(i)$ for the $m$-dimensional vector corresponding to the costs of item $i$, and we write $M_k$ for the $n$-dimensional vector, which is a single knapsack-constraint, corresponding to the $k^{\text{th}}$ row of $M$.  Likewise, for a set $X \subseteq [n]$, we write $M(X) = \sum_{i \in X} M(i)$. For an $n$-dimensional vector $y$, we write $M y \leq \vec 1$ to denote that  $M_k y \leq 1$ for all $k = 1, \dots, m$.

Given a vector $Y \in \{0,1 \}^n$ where the support of $Y$ is a $q$-additive pseudo-solution for a multi-knapsack constraint matrix $M$,  we sometimes say for brevity that $Y$ is a $q$-additive pseudo-solution for $M$, i.e. for each $k$ we can modify $Y$ to some $Y'$ by zeroing out at most $q$ entries such that $M_k Y' \leq \vec 1$.

We use Iverson notation where $[[\phi]]$ is the indicator function for a Boolean predicate $\phi$, i.e., $[[\phi]] = 1$ if $\phi$ is true, and $[[\phi]] = 0$ otherwise.  The $\tilde O()$ notation is defined as $\tilde O(x) = x \cdot \polylog(x)$.

A partition $\mathcal G$ of ground set $U$ is a collection of pairwise-disjoint sets $G_1, \dots, G_k$  with $U = G_1 \cup \dots \cup G_k$. We refer to  $G_i$ as the \emph{blocks} of the partition. For each $u \in U$, we define $\mathcal G(u)$  to be the unique block with $u \in \mathcal G(u) \in \mathcal G$. For $W \subseteq U$, we define $\mathcal G(W) \subseteq \mathcal G$ to be the set of blocks involved in $W$, i.e., $\mathcal G(W) = \{ \mathcal G(w) \mid w \in W \}$.   For a set of blocks $\mathcal D \subseteq \mathcal G$ and a set of items $W \subseteq U$, we define the \emph{restriction of $W$ to $\mathcal D$}, denoted $W \wedge \mathcal D$, to be the set of items in $W$ which are also part of a block of $\mathcal D$, i.e., $W \wedge \mathcal D =  \bigcup_{G \in \mathcal D} W \cap G =  \{w \in W \mid \mathcal G(w) \in \mathcal D \}$.

\subsection{Organization}
In Section~\ref{pseudo-add-sec1}, we develop a new  ``Samuels-Feige'' type of concentration inequality for unbounded random variables. This result is quite general, and may be of independent interest.   In Section~\ref{sec:bsl}, we develop and analyze the new KPR dependent rounding algorithm. Section~\ref{sec:varKPR} describes a few extensions of this algorithm. Sections~\ref{sec:knapsack-median}, \ref{sec:knapsack-median2}, \ref{sec:knapsack-center} describe the applications to knapsack median and knapsack center  problems. Section~\ref{advanced-cor-sec} provides further analysis of concentration bounds and near-independence properties of the KPR algorithm.

\section{A concentration inequality for additive knapsack pseudo-solutions}
\label{pseudo-add-sec1}
The main result of this section is to show an intriguing connection between independent rounding and pseudo-additive solutions. For maximum generality, we state it in terms of a broader class of random variables satisfying a property known as \emph{negative association} (NA) \cite{na-cite}, defined as follows:
 \begin{definition}[Negatively associated random variables \cite{na-cite}]
 Random variables $X_1, \dots, X_n$ are \emph{negatively associated} (NA) if for every subset $A \subseteq [n]$, and any pair of non-decreasing functions $f_1, f_2$, the random variables $f_1(X_i: ~i \in A)$ and $f_2(X_j: ~j \in [n] - A)$ have non-positive covariance.
(Here, ``$f_1(X_i: ~i \in A)$" means $f_1$ applied to the tuple $(X_i: i \in A)$, and similarly for ``$f_2(X_j: ~j \in [n] - A)$".)
 \end{definition}

If $X_1, \dots, X_n$ are independent random variables, then they are NA. The class of NA random variables includes other random processes; for example, the load-vector of the urns in balls-and-urns processes \cite{DBLP:journals/rsa/DubhashiR98}. 

 With this definition, we state our main result:
  \begin{theorem}
  \label{pseudo-add-ind}
  Let $X_1, \dots, X_n$ be negatively-associated,  non-negative random variables. Then with probability at least $1 - \delta$, there is a set $W \subseteq [n]$ (which may depend on the values of $X_1, \dots, X_n$) with $|W| \leq O( \sqrt{n \log \tfrac{1}{\delta}} )$ such that $\sum_{i \in [n] - W} X_i \leq \sum_{i \in [n]} \E[X_i]$.
  \end{theorem}

 Theorem~\ref{pseudo-add-ind} can be rephrased in the language of knapsack constraints:
  \begin{corollary}
  \label{pseudo-add-ind2}
 Let $X_1, \dots, X_n$ be negatively-associated,  non-negative random variables, and $a_1, \dots, a_n$ be non-negative coefficients with $\sum_i a_i \bE[X_i] = 1$.  Then with probability at least $1 - \delta$, the values $X_1, \dots, X_n$ form an $O(\sqrt{n \log \tfrac{1}{\delta}})$-additive pseudo-solution to the knapsack constraint $a_1 x_1 + \dots + a_n x_n \leq 1$.
  \end{corollary}
 
Notably, this bound does not depend on the variance of the variables $X_1, \dots, X_n$, which is quite different from conventional concentration bounds such as Chernoff-Hoeffding.  We remark that the bound is tight for many values of $n$ and $\delta$. For example, consider a system with $n$ independent Bernoulli$(p)$ variables for any constant $p \in (0,1)$. In this case, we have $|W| = \max\{  (\sum_i X_i) - n p , 0 \}$, which is on the order of $O( \sqrt{n \log \tfrac{1}{\delta} })$ with probability at least $1 - \delta$.

Let us briefly summarize the role of this concentration inequality in our overall rounding algorithm. As we have discussed, the KPR rounding algorithm stops with a vector $\tilde Y \in [0,1]^U$ with some $t$ remaining fractional entries, and which satisfies the knapsack constraints exactly. One attractive option to obtain a fully integral vector $Y$ is to apply independent selection to $\tilde Y$, ignoring the knapsack constraints. This  may violate the knapsack constraints, but by how much? This is precisely the random process governed by Corollary~\ref{pseudo-add-ind2}. The integral entries of $\tilde Y$ have no effect and so $\tilde Y$ effectively has $t$ variables. If we set $q = O(\sqrt{t \log m})$, then Corollary~\ref{pseudo-add-ind2} shows that $Y$ is a $q$-additive pseudo-solution to each knapsack constraint $M_k$ with probability $1 - \frac{1}{2m}$. A union bound over the $m$ knapsack constraints shows that $Y$ is a $q$-additive pseudo-solution for $M$ with constant probability.

 By way of comparison, \cite{feige2006sums,garnett:feige-bound,DBLP:journals/mor/HeZZ10,samuels1966chebyshev} give other concentration bounds for sums of nonnegative independent variables without regard to size or variance. Such results are useful for problems such as hypergraph matchings and probabilistic estimation of network parameters \cite{goldreich2008approximating,eden2016sublinear, DBLP:journals/jct/AlonHS12}. We anticipate that Theorem~\ref{pseudo-add-ind} may have broader applications beyond our rounding algorithm.
  
\subsection{Formal proof of Theorem~\ref{pseudo-add-ind}}
We begin by recalling some useful properties of NA random variables; see \cite{na-cite,DBLP:journals/rsa/DubhashiR98} for further details.
 \begin{proposition}
Let $X_1, \dots, X_n$ be NA random variables.
 \begin{description}
 \item[(Q1)] If $f_1, \dots, f_n$ are univariate non-decreasing functions, then $f_1(X_1), \dots, f_n(X_n)$ are NA as well.
 \item[(Q2)] If $X_1, \dots, X_n$ are bounded in the range $[0,1]$, then the Chernoff-Hoeffding bounds 
 (for both the upper- and lower-tails of sums) apply to them as they do to independent random variables.
 \end{description}
 \end{proposition}

The proof of Theorem~\ref{pseudo-add-ind} has two parts. First, we prove it under the assumption that the random variables $X_1, \dots, X_n$ have continuous cumulative density functions (CDF's); we then use a ``smoothing'' argument to extend it to arbitrary distributions.

Let us define $\lambda = \log \tfrac{1}{\delta}$.  If $n < c \lambda$ for any chosen constant $c$, then the result will hold trivially by taking $W = [n]$. Thus, we assume that $ n > c \lambda$ for any needed constant $c$ in the proof. 

\textbf{Part I.} Suppose that $X_1, \dots, X_n$ have continuous CDF's. By rescaling, we assume without loss of generality that $\sum_i \E[X_i] = 1$.  Since the CDF of the $X_i$ variables is continuous and the $X_i$ variables are non-negative, there is a real number $\alpha \geq 0$ such that
 \begin{equation}
 \label{lambda-eqn}
\sum_{i \in [n]} \Pr(X_i > \alpha) = 10 \lambda
\end{equation}

 We will take $W$ to be the set of all indices $i \in [n]$ with $X_i > \alpha$, thus $\E[ |W| ] = 10 \lambda$.  We need to show that $|W| \leq O(\lambda)$ and $\sum_{i \in [n] - W} X_i \leq 1$. Let us define the random variables $Y_i = \min(\alpha, X_i)$. Noting that $Y_i = X_i$ for $i \notin W$ while $Y_i = \alpha$ for $i \in W$, we see that the following equation holds with probability one:
\begin{equation}
\label{ghtt1}
\sum_{i \in [n] - W} X_i = \sum_{i \in [n]} \bigl( Y_i - \alpha [[ i \in W ]] \bigr)= - \alpha |W| + \sum_{i \in [n]} Y_i
\end{equation}

Let $\mathcal E_1$ denote the event that $5 \lambda \leq |W| \leq 20 \lambda$ and let $\mathcal E_2$ denote the event that $\sum_i Y_i \leq 1 + 5 \alpha \lambda$.  By Eq.~(\ref{ghtt1}), when $\mathcal E_1$ and $\mathcal E_2$ both hold, we have $\sum_{i \in [n] - W} X_i \leq 1$ and $|W| \leq O(\lambda)$, as is desired.  It remains to compute the probabilities of $\mathcal E_1$ and $\mathcal E_2$.  

First, since the function mapping $x$ to $[[x > \alpha ]]$ is non-decreasing, the indicator variables for $i \in W$ remain negatively associated by (Q1). By (Q2) and Eq.~(\ref{lambda-eqn}) we therefore have
  $$
 \Pr( \mathcal E_1) =  \Pr( 5 \lambda \leq |W| \leq 20 \lambda) \geq 1 - 2 e^{-\Omega(\lambda)} \geq 1 - 2 e^{-\Omega(\sqrt{n \log(1/\delta)})} \geq 1 - \delta/2
  $$
  using our assumption that $n > c \lambda$ for any sufficiently large constant $c$.

Next, since $Y_i \leq X_i$, we have $\E[ \sum_i Y_i] \leq 1$. The random variables $Y_i$ are all bounded in the range $[0, \alpha]$. Also  since the function mapping $x$ to $\min(x, \alpha)$ is non-decreasing, by (Q1) the random variables $Y_i$ remain negatively associated. By (Q2) we can apply Hoeffding's bound, and thus
$$
\Pr(\neg \mathcal E_2) = \Pr( \sum_{i \in [n]} Y_i > 1 +5 \alpha \lambda) \leq e^{-\frac{-2 (5 \alpha \lambda)^2}{n \alpha^2}} = e^{-50 \lambda^2/n} = e^{-50 \log(1/\delta)} \leq \delta/2
$$

Thus, we have $\Pr(\mathcal E_1 \cap \mathcal E_2) \geq 1 - \delta$, and in this case we have $\sum_{i \in [n] - W} X_i \leq 1$ for $|W| \leq O(\lambda)$. 
  
  \textbf{Part II.}  Now consider random variables $X_1, \dots, X_n$ (with no assumption on their CDF's). Let $Y_1, \dots, Y_n$ be independent random variables which are uniform in the range $[0, \tfrac{\epsilon}{n}]$, and let $Z_i = X_i + Y_i$. The random variables $Z_i$ are clearly non-negative, and they have continuous CDF's. For any fixed values of $Y_1, \dots, Y_n$, the random variables $Z_1, \dots, Z_n$ are NA by (Q1); this remains true after integrating over $Y_1, \dots, Y_n$.
  
 So as we showed in Part I of the proof, with probability at least $1- \delta$ we have
  $$
 \min_{\substack{ W \subseteq [n] \\ |W| \leq r}} \sum_{i \in [n] - W} Z_i  \leq \sum_{i \in [n]} \E[Z_i]
  $$
 for some parameter $r = O( \sqrt{n \lambda})$. Since $\E[Z_i] = \E[X_i] + \epsilon n / 2$ and $X_i \leq Z_i$, this implies that 
 $$
\Pr \Bigl( \min_{\substack{ W \subseteq [n] \\ |W| \leq r}} \sum_{i \in [n] - W} X_i \leq \sum_{i \in [n]} \E[X_i] + \epsilon/2 \Bigr) \geq 1 - \delta
$$
  
  Since this holds for every $\epsilon > 0$, this implies:
 \[
\Pr \Bigl( \min_{\substack{ W \subseteq [n] \\ |W| \leq r}} \sum_{i \in [n] - W} X_i \leq \sum_{i \in [n]} \E[X_i] \Bigr) \geq 1 - \delta
\]
  
This concludes the proof. We remark that our result is algorithmically ``local'' in that it discards each $i$ if and only if $X_i > \alpha$ for a certain threshold value $\alpha$; this automatically ensures that the discarded set is small with good probability.

\section{Knapsack-partition systems and rounding}
\label{sec:bsl}
The KPR rounding algorithm takes as input the partition $\mathcal G$ over ground set $U$, the constraint matrix $M$, the fractional vector $y$ satisfying $y(G) = 1$ for all $G \in \mathcal G$, and an integer parameter $t$. It returns a mostly rounded vector $\tilde Y$.  

For $G \in \mathcal G$, we define the following functions to count the number of fractional entries:
$$
T_G(y) = \max \bigl\{ 0, |\{i \in G \mid y_i \in (0,1) \} | - 1 \bigr \}, \qquad \text{and} \qquad T(y) = \sum_{G \in \mathcal G} T_G(y)
$$

We then define the algorithm as follows:
\begin{algorithm}[H]
\caption{$\textsc{KPR}(\mathcal G, M, y,t)$}
\begin{algorithmic}[1]
\FOR{each block $G \in \mathcal G$}
\STATE Execute an unbiased walk  to generate random vector $y' \in [0,1]^U$ with $\E[y'] = y$, and such that $y'$ is an extreme point of the polytope: $ \Bigl \{  My' =  My, y'(G) = 1, y'_j = y_j \text{ for $j \notin G$} \Bigr \}$.
\STATE Update $y \leftarrow y'$
\ENDFOR
\WHILE{$T(y) > t$} 
\STATE Form a set $\mathcal J \subseteq \mathcal G$, wherein each $G \in \mathcal G$ goes into $\mathcal J$ independently with probability $p = 3m/T(y)$.
	\IF{$\sum_{G \in \mathcal J} T_G(y) \geq m+1$}
		\STATE Choose $\delta \in \mathbb{R}^U$ such that 
		\begin{itemize}[nolistsep]
			\setlength{\itemindent}{+.2in}
			
			\item \hspace{-0.3in} $M\delta = 0$, $y+\delta\in [0,1]^U$, and $y-\delta \in [0,1]^U$
			\item \hspace{-0.3in} There is at least one index $i$ with $y_i \in (0,1)$ such that $y_i+\delta_i \in \{0, 1 \}$ or $y_i-\delta_i \in \{0, 1 \}$.
			\item  \hspace{-0.3in} $\delta_j = 0$ if $\mathcal G(j) \notin \mathcal J$
			\item \hspace{-0.3in} $\delta(G) = 0$ for all $G \in \mathcal G$.
		\end{itemize}
		\STATE With probability $1/2$, update $y \gets y + \delta$; else, update $y \gets y - \delta$
	\ENDIF
	\ENDWHILE
\RETURN $y$
\end{algorithmic} 
\label{algo:depround}
\end{algorithm}

The loop at lines 1--3 is a preprocessing step consisting of straightforward dependent rounding within each block $G$. We write $y' =   \textsc{IntraBlockReduce}(y)$ for the vector obtained at the termination of the loop. After this step, the algorithm repeatedly applies a more-complicated rounding process which modifies multiple blocks. We write $y' = \textsc{KPR-iteration}(y)$ to denote a single iteration of the loop at lines 4--8. Thus, the overall algorithm is equivalent to the following: 

\begin{algorithm}[H]
\caption{$\textsc{KPR}(\mathcal G, M, y,t)$, summarized}
\begin{algorithmic}[1]
\STATE $y \leftarrow \textsc{IntraBlockReduce}(y)$ 
\WHILE{$T(y) > t$}
\STATE update $y \leftarrow \textsc{KPR-iteration}(y)$
\ENDWHILE
\RETURN $y$
\end{algorithmic} 
\label{algo:deprounda1}
\end{algorithm}

The KPR algorithm \emph{requires throughout} that $t > 12 m$; this assumption will not be stated explicitly again. Because of this condition, the probability $p$ in line 5 is at most $3 m/t \leq 1/4$.  (Note that it is likely impossible to obtain fewer than $m$ fractional entries, while still respecting the knapsack constraints.) 

Also, although budget matrices for clustering problems are usually assumed to be non-negative, we dol not require this for KPR. The entries of the matrix $M$ can be arbitrary real numbers. 

\subsection{KPR algorithm: formal results}
As we have discussed, desiderata (D1)--(D4) cannot be exactly satisfied.  To describe the negative correlation properties of KPR, we use a potential function $Q(W,x)$ defined for a set $W \subseteq U$ and vector $x \in [0,1]^U$ as follows:
  $$
  Q(W,x) = \prod_{G \in \mathcal G}  (1 - x(W \cap G)).
  $$

 We will show that the vector $\tilde Y = \textsc{KPR}(\mathcal G, M, y, t)$ satisfies the following constraints:
\begin{enumerate}
\item[(E1)] For all $W \subseteq U$,
$\E[Q(W,\tilde Y)]$ is ``not much more than'' $Q(W,y)$;
\item[(E2)] Every $j \in U$ has $\E[\tilde Y_j] = y_j$;
\item[(E3)] $\tilde Y(G) = 1$ for $G \in \mathcal G$;
\item[(E4)] $M \tilde Y = M y$;
\item[(E5)] At most $2 t$ entries of $\tilde Y$ are fractional.
\item[(E6)] For each block $G \in \mathcal G$, at most $m+1$ entries of $\tilde Y$ are fractional.
\end{enumerate}

 (E1) is intentionally vague, as the relationship between $\E[Q(W,\tilde Y)]$ and $Q(W,y)$ is quite complex.  Our main result covers a setting needed for a number of our clustering algorithms, where there is a relatively small set $\mathcal D$ of blocks $G$ which have $y(G \cap W)$ close to one. Formally, we show the following:
\begin{theorem}
  \label{knap-KPR-thm}
  Let $\mathcal D \subseteq \mathcal G$ with $|\mathcal D| = d$ and let $\tilde Y = \textsc{KPR}(\mathcal G, M, y, t)$ with  $t \geq 5000 m (d+1)$. Then for any set $W \subseteq U$, there holds
  $$
  \E[Q(W,\tilde Y)] \leq Q(W,y) + Q(W \wedge \mathcal D,y) \bigl( e^{O( (d+1)^2 m^2/t) } - 1 \bigr)
  $$
\end{theorem}

Recall that we define $W \wedge \mathcal D =  \bigcup_{G \in \mathcal D} W \cap G$.   Theorem~\ref{knap-KPR-thm} is complex and hard to use directly. We derive a number of simplified results, such as the following three estimates:
\begin{theorem}
  \label{final-thm2}
   Let $\tilde Y = \textsc{KPR}(\mathcal G, M, y, t)$ and let $W \subseteq U$.
 \begin{enumerate}
\item[(a)] For $\mathcal D \subseteq \mathcal G$ with $t > 5000 m^2 (|\mathcal D|+1)^2$  there holds $\E[Q(W,\tilde Y)] \leq Q(W,y) + O( (|\mathcal D|+1)^2 m^2/t ) \cdot Q(W \wedge \mathcal D,y).$
 \item[(b)] For $t > 12 m$,  there holds $\E[Q(W,\tilde Y)] \leq Q(W,y) + O( m^2/t )$.
 \item[(c)] For $t > 10000 m d$ where $d = |\mathcal G(W)|$,  there holds  $\E[Q(W,\tilde Y)] \leq Q(W,y) e^{O(m^2 d^2/t)}$.
 \end{enumerate}
\end{theorem}

There are two key steps in \textsc{KPR} to ensure property (E1).  First, the modification vector $\delta$ in line 4 is always bounded by the current value of $y$. This ensures that the typical change in the value of $Q(W,y)$ is proportional to the current value of $Q(W,y)$. Second, the set $\mathcal J$, which determines the entries of $y$ to modify, is randomly selected.  This spreads out the (inevitable) correlation among the entries of $y$. 

 Before we show Theorem~\ref{knap-KPR-thm}, let us explain the role played by the potential function $Q$. Observe that for $Y = \textsc{IndSelect}(\mathcal G, y)$ we have $Y(W) = 0$ if and only if $Q(W,Y) = 1$ and so 
$$
\Pr(Y(W) = 0) = \bE[ Q(W, Y) ] = Q(W,y);
$$
thus, $Q(W,\tilde Y)$ is a smoothed measure of whether the KPR output $\tilde Y$ satisfies $\tilde Y(W) = 0$.

Why might one be interested in upper-bounding terms of the form $\Pr(Y(W) = 0)$? We have briefly touched on this, but let us spell out in greater detail  how such bounds arise in clustering algorithms, such as our algorithms for knapsack median and knapsack center. The simplest versions of these algorithms first cluster the facilities in some greedy manner; these are the blocks $G$ of the partition. They then open a facility suitably at random from each block. Any given client $j$ will first check if some ``nearby'' facility gets opened; if not, then it must use a ``backup'' facility which, however, is farther away. The bad event of no opened ``nearby'' facility corresponds to  $Y(W) = 0$, where $Y$ is the indicator vector for which facilities are open and $W$ is the set of nearby facilities.

In more advanced algorithms, multiple facilities may be opened from a cluster, or the clusters may have even more complex interactions. These cases can also be interpreted as knapsack-partition systems, and again the distance for a client $j$ can be recast in terms of events of the form $Y(W) = 0$.

Section~\ref{advanced-cor-sec} shows some further upper bounds on $\bE[Q(W,\tilde Y)]$. These do not follow from Theorem~\ref{knap-KPR-thm}, and are not directly useful for our clustering algorithms. For simplicity, we do not attempt to optimize the constant factors here or elsewhere in the analysis.  
    
We will begin by showing some easier properties of this algorithm, including that it is well-defined and terminates in polynomial time. The proof of (E1), which is much harder, comes next. 

\subsection{Simple properties and convergence of \textsc{KPR}} 
\begin{proposition}
\label{igprop}
The vector $y' = \textsc{IntraBlockReduce}(y)$ satisfies properties (E2), (E3), (E4), (E6). Furthermore, for any $W \subseteq U$ we have $\bE[Q(W, y')] = Q(W,y)$.
\end{proposition}
\begin{proof}
The polytope in line 2 of \textsc{KPR} has $m+1$ constraints among entries $y_j$ for $j \in G$, so an extreme point has at most $m+1$ fractional entries. The polytope conditions preserve properties (E3) and (E4), and preserve (E2) since the walk is unbiased. Finally, the change in $y$ during each iteration is confined to block $G$. Since $Q(W,y)$ is linear function of $y_j$ for $j \in G$, (E2) ensures that  $\E[Q(W,y)]$ does not change. 
\end{proof} 

\begin{proposition} 
\label{prog-prop2}The vector $\delta$ in line 7 of \textsc{KPR} exists and can be found efficiently.
\end{proposition}
\begin{proof}
In forming the vector $\delta$, there is one degree of freedom for each entry $i \in G$ with $y_i \in (0,1)$ and $G \in \mathcal J$. We further have $m$ linear constraints (from the matrix $M$) and $|\mathcal J|$ linear constraints (from the condition $\delta(G) = 0$ for each $G \in \mathcal J$). This gives a total of $\sum_{G \in \mathcal J} T_G(y) - m$ degrees of freedom. 

So the linear system has a non-zero solution vector $v$ as long as $\sum_{G\in \mathcal J} T_G(y) \geq m+1$, which is precisely the condition at line 6 of \textsc{KPR}. Now choose $a \in \mathbb R$ to be maximal such that $y + a v \in [0,1]^U$ and $y - a v \in [0,1]^U$. One may verify that $a < \infty$ and setting $\delta = a \gamma$ achieves the claimed result. 
\end{proof}

\begin{proposition}
\label{prog-prop}
In each iteration $y' \leftarrow \textsc{KPR-iteration}(y)$ of \textsc{KPR}, there is a probability of at least $0.24$ that $y'$ has at least one more integral coordinate than $y$.
\end{proposition}
\begin{proof}
By Proposition~\ref{igprop}, the vector $y$ after \textsc{IntraBlockReduce} has $T_G(y) \leq m$ for each $G \in \mathcal G$. We also have $T(y) > t$, as otherwise \textsc{KPR} would have terminated.

For each $G \in \mathcal G$, define $Z_G = [[G \in \mathcal J]] T_G(y)$ and define $Z = \sum_G Z_G$.  If $Z \geq m+1$, then there is at least $1/2$ probability of producing at least one new rounded entry in $y'$. Also, $Z$ is a sum of independent random variables with mean $p T(y) = 3m$; furthermore, each $Z_G$ is bounded in the range $[1,m]$. A simple analysis with Chernoff's bound shows that $\Pr( Z \geq m+1) \geq 0.48$, which gives the claimed result.
\end{proof}

We will later need  the following stronger version of Proposition~\ref{prog-prop} where we condition on a given iteration not touching a given subset of blocks.
\begin{proposition}
  \label{ttr1prog}
  Consider an iteration $y' \leftarrow \textsc{KPR-iteration}(y)$ of \textsc{KPR}. For any set of blocks $\mathcal D \subseteq \mathcal G$ with  $t \geq 5000 m (|\mathcal D|+1)$, there is a probability of at least $1/10$ that $\mathcal D \cap \mathcal J = \emptyset$ and $T(y') < T(y)$. 
\end{proposition}
\begin{proof}
 Let $d = |\mathcal D|$. Letting $\mathcal E$ denote the event that $\mathcal D \cap \mathcal J = \emptyset$, we have
  $$
  \Pr(\mathcal E) =  (1 - 3 m/T(y))^d \geq (1 - 3 m/t)^d \geq \Bigl( 1 - \frac{3 m}{5000 m (d+1)} \Bigr)^d \geq 0.99
$$
  
  Next, as in Proposition~\ref{prog-prop}, let us define $Z = \sum_{G \in \mathcal J} [[ G \in \mathcal J ]] T_G(y)$. If $Z \geq m+1$, then with probability $1/2$ there will be at least one rounded variable. Furthermore, conditional on event $\mathcal E$, here $Z$ is a sum of independent random variables in the range $[1,m]$ and with mean $\mu = \sum_{G \in \mathcal G - \mathcal D} p T_G(y)$. Since $T(y) \geq t$ and $T_G(y) \leq m$ for all blocks $G$, we have $\mu \geq (3 m/t) \cdot (t - d m) \geq \frac{3m (5000 m (d+1) - d m)}{5000 m d}  \geq 2.99 m$.  So, by Chernoff's bound, $\Pr(Z < m+1 \mid \mathcal E) \leq e^{-(2.99)^2 (0.34)^2/2} \leq 0.60$.  Overall, the desired event happens with probability at least $0.99 \cdot (1 - 0.60) \cdot 1/2 \geq 1/10$.
\end{proof}

\begin{proposition} 
\label{proposition:depround2}
The output $\tilde Y$ of $\textsc{KPR}(\mathcal G, M, y,t)$ satisfies properties (E2) --- (E6).
\end{proposition}
\begin{proof}
 By Proposition~\ref{igprop}, the conditions (E2), (E3), (E4), (E6) hold after \textsc{IntraBlockReduce}. Each application of \textsc{KPR-iteration} updates $y \leftarrow y \pm \delta$. Since $\delta(G) = 0$ and $M \delta = 0$, properties (E3), (E4) are preserved. The expected change in $y$ is $\tfrac{1}{2} \delta + \tfrac{1}{2} (-\delta) = 0$, so property (E2) is preserved.  For property (E5), note that \textsc{KPR} terminates when $T(y) \leq t$, in which case $\tilde Y = y$ has at most $2 t$ fractional entries.
\end{proof}

\begin{proposition} \textsc{KPR} runs in expected polynomial time.
\end{proposition}
\begin{proof}
Clearly \textsc{IntraBlockReduce} runs in polynomial time. By Proposition~\ref{prog-prop}, each iteration of \textsc{KPR-iteration} has a probability $\Omega(1)$ of causing a new entry to become integral, in which case $T(y)$ decreases by at least one.  This implies that the expected number of iterations is $O(|U|)$, and by Proposition~\ref{prog-prop2}, each iteration can be implemented in polynomial time.
\end{proof}

\subsection{Property (E1): Proof of Theorem~\ref{knap-KPR-thm}}

Before the formal analysis, let us provide an overview. Consider the evolution of the random variable $Q(W,y)$ for some set $W$. There may be small increase in the expected value of $Q(W,y)$ in each iteration of KPR, but this is compensated by steady decrease in $T(y)$.  To measure this, we will define a potential function $\Phi$ which depends on $Q(W,y)$ and $T(y)$; we will show that $\E[\Phi(y)]$ does not increase in any iteration of \textsc{KPR}.  Furthermore, when $T(y) \leq t$ at the end of the process, we have $\Phi(y) = Q(W,y)$. Consequently, the final value $\bE[Q(W,\tilde Y)] = \E[\Phi(\tilde Y)]$ is at most the initial value $\Phi(y)$.

We first show a useful result on how the potential function $Q$ changes during a single iteration of KPR.
\begin{lemma}
  \label{prop2a}
Suppose we are in the middle of executing \textsc{KPR} with state vector $y$, and let $y'$ be state vector at the next iteration. If we condition on the vector $y$,  then for each $W \subseteq U$ there holds
  $$
  \E[Q(W, y') \mid y] \leq Q(W,y) \cosh \Bigl( \frac{6 m}{T(y)} \sum_{G \in \mathcal G} y(G \cap W)  \Bigr) 
  $$
\end{lemma}
\begin{proof}
Define $S = Q(W, y)$ and $S' = Q(W, y')$.  All probability calculations here are conditioned on $y$.  Let us first condition as well on the random variable $\delta$. Note $y' = y + \delta$ or $y' = y - \delta$, each with probability $1/2$. If we define $b_G = y(G \cap W)$ and $\mu_G = \delta( G \cap W)$ for each $G \in \mathcal G$, we therefore get:
\begin{align*}
\E[S' \mid \delta] &= 1/2 \prod_{G} (1 - b_G - \mu_G) + 1/2 \prod_{G} (1 - b_G + \mu_G) \\
&= 1/2 \sum_{X \subseteq \mathcal G} \prod_{G \in X} (- \mu_G) \prod_{G \in U - X} (1 - b_G) + 1/2 \sum_{X \subseteq \mathcal G} \prod_{G \in X} (\mu_G) \prod_{G \in U - X} (1 - b_G) \\
&= \sum_{\substack{X \subseteq \mathcal G \\ \text{$|X|$ even}}} \prod_{G \in X} \mu_G \prod_{i \in U - X} (1 - b_G) = S \sum_{\substack{X \subseteq \mathcal G\\ \text{$|X|$ even}}} \prod_{G \in X} \frac{\mu_G}{1 - b_G} \leq S \sum_{\substack{X \subseteq \mathcal G\\ \text{$|X|$ even}}} \prod_{G \in X} \frac{|\mu_G|}{1 - b_G}
\end{align*}

For each block $G$ define $R_G = 2 [[G \in \mathcal J ]] b_G(1-b_G)$. We claim that $|\mu_G| \leq R_G$ with probability one. For, if  $G \not \in \mathcal J$, then $\mu_G = 0 = R_G$. If $G \in \mathcal J$, then necessarily $\sum_{j \in G \cap W} (y_j + \delta_j)$ is bounded in $[0,1]$, and hence $|\mu_G| \leq \min(b_G, 1-b_G) \leq 2 b_G (1-b_G) = R_G$. Therefore, 

$$
\E[S' \mid \delta] \leq S \sum_{\substack{X \subseteq \mathcal G\\ \text{$|X|$ even}}} \prod_{G \in X} \frac{|\mu_G|}{1 - b_G} \leq S \sum_{\substack{X \subseteq \mathcal G\\ \text{$|X|$ even}}} \prod_{G \in X} \frac{R_G}{1 - b_G}
$$

The random variables $R_G$ are independent, and each has mean $\E[R_G] = (3 m/T(y)) \cdot 2 b_G (1-b_G)$. Now integrate over random variables $\delta$ and $\mathcal J$  to obtain:
\begin{align*}
\E[S'] &\leq  S \sum_{\substack{X \subseteq \mathcal G\\ \text{$|X|$ even}}} \prod_{G \in X} \frac{\E[R_G]}{1 - b_G} \leq S \sum_{v=0}^{\infty} \frac{1}{ (2 v)!}  \left(\sum_{G} \frac{\E[R_G]}{1 - b_G} \right)^{2 v} = S \cosh\left( \sum_G \frac{\E[R_G]}{1 - b_G} \right) \\
&\leq S \cosh \left( \sum_{G}  \frac{ (3 m/T(y)) \cdot 2 b_G (1-b_G)}{1 - b_G} \right)  = S \cosh\left(\frac{6 m \sum_G b_G}{T(y)} \right) \qedhere
\end{align*}
\end{proof}

We are now ready to show Theorem~\ref{knap-KPR-thm}; Theorem~\ref{final-thm2} will follow as an immediate corolalry.
\begin{proof}[Proof of Theorem~\ref{knap-KPR-thm}]
Let us fix $\mathcal D \subseteq \mathcal G$ and $W \subseteq U$ with $|D| = d$ and $t \geq 5000 m (d+1)$. We define parameter $a = 8000 m^2 (d+1)^2$ and we define potential function $\Phi$ for a vector $x \in [0,1]^U$ by:
$$
\Phi(x) =
\begin{cases}
Q(W, x) +  Q(W \wedge \mathcal D, x) \bigl( e^{a(\frac{1}{t} - \frac{1}{T(x)})}  - 1\bigr) & \text{if $T(x) \geq t$} \\
Q(W, x) & \text{if $T(x) < t$} 
\end{cases} 
$$

The key to the proof is to show that, if we are in the middle of executing \textsc{KPR}, with state vector $y$, and we let $y'$ be the state vector at the next iteration, then there holds
\begin{equation}
\label{hheqn15}
\E[ \Phi(y') \mid y] \leq \Phi(y)
\end{equation}

To show Eq.~(\ref{hheqn15}), let us fix the state $y$. We may assume that $T(y) > t$, as otherwise the algorithm is done and $y = y'$. We define a number of parameters as follows:
\begin{eqnarray*}
&\begin{split} 
S_0 &= Q(W \wedge \mathcal D, y) \\
 S_1 &= Q(W \wedge (\mathcal G - \mathcal D), y) \\
S &= Q(W, y) = S_0 S_1 \\
 T &= T(y) \\
 \beta &= e^{a(1/t - 1/T)}
 \end{split}
 \qquad
 \begin{split}
 S_0' &= Q( W \wedge \mathcal D, y')  \\
 S'_1 &= Q(W \wedge (\mathcal G - \mathcal D), y') \\
S' &= Q(W, y') = S_0' S_1' \\
 T' &= T(y') \\
 \beta' &= \max\{ 1, e^{a(1/t - 1/T')} \}
 \end{split}
 \end{eqnarray*}
 
With this notation, we observe that $$
\Phi(y) = S + S_0 (\beta - 1), \qquad \Phi(y') = S' + S_0' (\beta' - 1).
$$

 Now let $\mathcal E_1$ denote the event that $T' < T$; when $\mathcal E_1$ occurs, then we have $\beta' \leq e^{a(1/t - 1/(T-1))}$ since $T > t$. Our condition on $t$ implies that $T \geq t \geq 2 \sqrt{a}$;  as we show in Proposition~\ref{tech-prop4}, we thus have 
$$
\beta' \leq \beta - \frac{a [[\mathcal E_1]] \beta}{2 T^2}.
$$
So $\Phi (y') \leq S' + S_0' \bigl( \beta-1 - \frac{a [[\mathcal E_1]] \beta}{2 T^2} \bigr)$. Taking expectations gives
$$
  \E[\Phi(y')] - \Phi(y) \leq \E[ S' ] + \E [ S_0' ( \beta-1) ] - \E \Bigl[ S_0' \frac{a [[\mathcal E_1]] \beta}{2 T^2} \Bigr] - S - S_0(\beta-1)$$
  which we rearrange as:
  \begin{equation}
  \label{sseqn}
  \E[\Phi(y')] - \Phi(y) \leq  \E[ S' - S] + (\beta-1) \E[S_0' - S_0] - \frac{a \beta \E[ S_0' \mid \mathcal E_1 ] \Pr( \mathcal E_1)}{2 T^2} 
  \end{equation} 

Let us consider the terms in Eq.~(\ref{sseqn}) in turn.  For the last term, define $\mathcal E_2$ to be the event that the set $\mathcal J$ formed at line 5 of KPR satisfies $\mathcal D \cap \mathcal J = \emptyset$. We can estimate $\E[S_0' \mid \mathcal E_1] \Pr(\mathcal E_1) \geq  \E[S_0' \mid \mathcal E_1 \cap \mathcal E_2] \Pr( \mathcal E_1 \cap \mathcal E_2 )$. When $\mathcal E_1$ and $\mathcal E_2$ occur, then $S_0' = S_0$ as none of the entries in $W \wedge \mathcal D$ are modified. By Proposition~\ref{ttr1prog}, $\Pr( \mathcal E_1 \cap \mathcal E_2 ) \geq 1/10$ so overall $\E[S_0' \mid \mathcal E_1] \Pr( \mathcal E_1 ) \geq S_0/10$.  Substituting the value of $a$, we get
\begin{equation}
\label{rle1}
 \frac{a \beta \E[ S_0' \mid \mathcal E_1 ] \Pr( \mathcal E_1)}{2 T^2}  \geq 400 m^2 (d+1)^2 \beta S_0 / T^2
 \end{equation}
Our next step is to estimate the term $\E[S_0' - S_0] $.  By Lemma~\ref{prop2a} applied to the set $W \wedge \mathcal D$, we have
 $$
  \E[S_0'] \leq S_0 \cosh \Bigl(6 m \sum_{G \in \mathcal D} y(G \cap W) / T \Bigr) \leq S_0 \cosh \bigl(6 m d / T \bigr)
  $$
  
  By our assumption that $T \geq t \geq 5000 m (d+1)$, this is at most $S_0 (1 + (6 m d/T)^2)$, and so
  \begin{equation}
  \label{rle2}
  \E[S_0'] - S_0 \leq 36 S_0 (m d/T)^2 \leq 36 S_0 m^2 (d+1)^2 / T^2
  \end{equation}

  Finally, we turn to estimating $\E[S' - S]$. By Lemma~\ref{prop2a} applied to the set $W$, we have $\bE[S'] \leq S \cosh \bigl( 6 m \sum_G y(G \cap W) / T \bigr)$. Since $0 \leq y(G \cap W) \leq 1$ and $x \leq -\ln(1-x)$ for all $x \in [0,1]$, we  have:
  \begin{align*}
\bE[S'] \leq S \cosh \Bigl(\frac{6 m}{T} \bigl( \sum_{G \in \mathcal D} 1 - \sum_{G \in \mathcal G - \mathcal D}  \ln(1 - y(G \cap W)) \bigr) \Bigr) = S_0 S_1 \cosh \bigl(6 m ( d - \ln S_1 ) /  T \bigr)
    \end{align*}
    
    As we show in Proposition~\ref{tech-prop3}, this implies that
  \begin{equation}
  \label{rle3}
  \E[S'] - S \leq  S_0 \Bigl( S_1 \cosh \bigl ( 6 m (d - \ln S_1) / T \bigr) - S_1 \Bigr) \leq S_0 \bigl( 6 m (d+1)/T  \bigr)^2
  \end{equation}

  Substituting the estimates of Eqs.~(\ref{rle1}), (\ref{rle2}), (\ref{rle3}) into Eq.~(\ref{sseqn}), we see that
  \begin{align*}
    \E[\Phi(y')] - \Phi(y) &\leq 36 S_0 m^2 (d+1)^2 / T^2 + 36 (\beta-1) S_0 m^2 (d+1)^2/T^2 - 400 m^2 (d+1)^2 S_0 / T^2 \\
    &= S_0 m^2 (d+1)^2 \bigl( 36 + 36 (\beta - 1) - 400 \beta \bigr) / T^2 = - 364 S_0 m^2 (d+1)^2 / T^2
  \end{align*}
  which is non-positive;  this shows Eq.~(\ref{hheqn15}) as desired.
  
  To complete the proof of Theorem~\ref{knap-KPR-thm}, suppose we execute KPR with input vector $y$. Let $y'$ be the vector after \textsc{IntraBlockReduce} and let $\tilde Y$ be the output vector.  Since $\E[Q(X,y')] = \bE[Q(X, y)]$ for all sets $X \subseteq U$, we have
$$
\E[\Phi(y')] \leq \E[Q(W, y')] + (e^{a/t} - 1) \E[Q(W \wedge \mathcal D, y')] = Q(W, y) + (e^{a/t} - 1) Q(W \wedge \mathcal D, y)
$$
By Eq.~(\ref{hheqn15}), and induction on all iterations of KPR, the output $\tilde Y$ satisfies $\E[ \Phi(\tilde Y)  ] \leq \E[\Phi(y') ]$. At the termination of \text{KPR}, we have $T(\tilde Y) \leq t$ and so $\Phi(\tilde Y) = Q(W, \tilde Y)$. Putting these inequalities together, we have shown that
\[
\E[Q(W,\tilde Y)] \leq Q(W, y) + Q(W \wedge \mathcal D, y) ( e^{a/t} - 1)~. \qedhere
\]
\end{proof}

\begin{proof}[Proof of Theorem~\ref{final-thm2}]
For part (a), observe that when $t \geq 5000 m^2 (d+1)^2$ we have  $e^{O(m^2 (d+1)^2/t)}-1 = O(m^2 (d+1)^2/t)$.  Part (b) follows from  part (a) with $\mathcal D = \emptyset$; note that if $t < 10000 m^2$ then the bound holds vacuously  since $Q(W,\tilde Y) \leq 1$ with probability one.   Part (c) follows Theorem~\ref{knap-KPR-thm} with $\mathcal D = \mathcal G(W)$; note that if $d = 0$, then $W = \emptyset$ so the bound holds vacuously.
\end{proof}

\section{Variants of KPR}
\label{sec:varKPR}
We summarize here some simpler ways to use KPR, which will occur in a number of algorithmic scenarios.

\subsection{KPR followed by independent selection}
One natural rounding strategy for a knapsack-partition problem is to execute KPR up to some stage $t$, and then finish by independent rounding.  We define this formally as the algorithm \textsc{FullKPR}:
\begin{algorithm}[H]
  \caption{$\textsc{FullKPR}(\mathcal G, M, y, t)$}
  \begin{algorithmic}[1]
    \STATE $\tilde Y \gets \textsc{KPR}(\mathcal G,  M, y,t)$
    \STATE $Y \gets \textsc{IndSelect}(\mathcal G, \tilde Y)$
    \RETURN $Y$
  \end{algorithmic}
  \label{algo:fullKPR}
\end{algorithm}

The resulting vector $Y \in \{0,1 \}^U$ is fully integral; it will not exactly satisfy the knapsack constraints, but it will be relatively close (depending on the value of $t$). Since independent selection does not change the expectation of $Q(W,y)$, all the analysis for \textsc{KPR} carries over immediately to \textsc{FullKPR}.

\begin{theorem}
  \label{full-KPR-thm1}
    Let $Y = \textsc{FullKPR}(\mathcal G, M, y,t)$ with $t > 12 m$. Suppose that the constraint matrix $M$ is non-negative and satisfies $M y \leq \vec 1$. Then with probability at least $1 - \delta$, the vector $Y$ is a $q$-additive pseudo-solution to $M$ where  $q = O(\sqrt{t \log \tfrac{m}{\delta}})$. This probability bound holds even after conditioning on the fixed vector $\tilde Y$.
\end{theorem}
\begin{proof}
Let us fix $\tilde Y = \textsc{KPR}(\mathcal G, M, y,t)$, and let $\mathcal G'$ denote the set of blocks where $\tilde Y$ has a fractional entry.   By Proposition~\ref{proposition:depround2}, we have $M \tilde Y = M y \leq \vec 1$, and $|\mathcal G'| \leq 2 t$.

Consider some row $k$ of the constraint matrix. Define the random variable $Z_G = \sum_{j \in G} M_k(j) Y_j$ for each block $G \in \mathcal G$, and note that  $M_k Y = \sum_{G \in \mathcal G} Z_G$.  By  Corollary~\ref{pseudo-add-ind2}, with probability of $1 - \delta/m$ there is a subset of blocks $\mathcal G'_k \subseteq \mathcal G'$ with $|\mathcal G'_k| \leq q=O( \sqrt{t \log \tfrac{m}{\delta}})$, such that $\sum_{G \in \mathcal G' -  \mathcal G'_k} Z_G \leq \sum_{G \in \mathcal G'} \bE[Z_G]$.  The vector $Y$ has one non-zero entry per block, and if we zero out the entries in the blocks of $\tilde G'_k$, the resulting vector $Y'$ has $M_k Y' \leq \sum_{G \in \mathcal G - \mathcal G'_k} Z_G \leq \sum_{G \in \mathcal G} \bE[Z_G] = M_k \tilde Y = M_k y$.  Thus $Y$ is a $q$-additive pseudo-solution for $M_k$ with probability at least $1 - \delta/m$.  Now take a union bound over all $m$ rows.
\end{proof}

  \begin{proposition}
    \label{lprop1}
    Let $Y = \textsc{FullKPR}(\mathcal G, M, y,t)$ and let $W \subseteq U$.
    \begin{enumerate}
    \item For $t > 12 m$, there holds $\bE[ \prod_{j \in W} Y_j] \leq O(m^2/t) + \prod_{j \in W} y_j$.
\item  For $t > 10000 m |W|$, there holds $\bE[ \prod_{j \in W} Y_j] \leq e^{O(m^2 |W|^2/t)} \prod_{j \in W} y_j$.
\end{enumerate}
  \end{proposition}
  \begin{proof}
We assume that the elements of $W$ all come from distinct blocks, as otherwise the LHS is zero and this holds immediately. Consider  the set $W' = \{ j \in U - W : \mathcal G(j) \in \mathcal G(W) \}$. We have $|\mathcal G(W')| \leq |\mathcal G(W)| = |W|$ and $\bE[ \prod_{j \in W} Y_j ] = Q(W, Y)$. Finally, note that $\bE[ Q(W, Y) ] = \bE[Q (W, \tilde Y)]$ where $\tilde Y$ is the vector at line 1 of \textsc{FullKPR}. Now apply Theorem~\ref{final-thm2}.
  \end{proof}
  
\subsection{Dependent rounding for knapsack constraints}
\label{dep-round-sec}
Given a vector $x \in [0,1]^U$ and a multi-knapsack constraint $M$ with $M x \leq \vec 1$, the problem of \emph{knapsack rounding} is to produce an integral vector $X \in \{0,1 \}^U$ which (as closely as possible) satisfies the knapsack constraint $M X \leq \vec 1$ and has probabilistic properties related to vector $x$ such as $\bE[X] = x$ coordinatewise. Note that this includes cardinality rounding as a special case, with $m = 1$ and $M(i) = 1/r$ for all $i \in U$.

We can interpret knapsack rounding as a special case of KPR. To do so, we extend the set $U$ to a larger ground-set $\overline U$; for each item $i \in U$, we have a corresponding ``dummy'' item $\bar i$. We then form a vector $y \in [0,1]^{\overline U}$ by setting $y_i = x_i$ and $y_{\bar i} = 1 - x_i$ for each $i \in U$, and we define a block $G_i = \{i, \bar i \}$. We lift the knapsack constraints $M$  to $\overline U$ by setting $M(\bar i) = 0$ for all $i$. 

We can then run KPR on this resulting knapsack-partition instance, and return the fractional vector $\tilde X \in [0,1]^U$ defined as $\tilde X_i = \tilde Y_i$ for $i \in U$.  We let $\tilde X = \textsc{KnapRound}(x,M,t)$ be the result of this process. Note that $M \tilde X = M \tilde Y = M y$, and that $\tilde X$ has at most $2 t$ fractional entries.  We can state a particularly crisp form of our near negative-correlation bounds in this setting:
\begin{proposition}
  \label{KPR-dep-round-prop}
Let $\tilde X = \textsc{KnapRound}(x,M,t)$, and let $S, T$ be disjoint subsets of $U$. Let $d = |S \cup T|$.
  \begin{enumerate}
  \item For $t > 12 m$ there holds $
  \E \bigl[ \prod_{i \in S} \tilde X_i \prod_{i \in T} (1 - \tilde X_i) \bigr] \leq  O(m^2/t) + \prod_{i \in S} x_i \prod_{i \in T} (1 - x_i)
  $
  \item For $t > 10000 m d$ there holds  $  \E \bigl[ \prod_{i \in S} \tilde X_i \prod_{i \in T} (1 - \tilde X_i) \bigr] \leq e^{O(m^2 d^2 / t)} \cdot \prod_{i \in S} x_i \prod_{i \in T} (1 - x_i)$
  \end{enumerate}
\end{proposition}
\begin{proof}
Define $W \subseteq \overline U$ by $W = \{\bar i \mid i \in S \} \cup \{ i \mid i \in T \}$.  Clearly $|W| = d$. Also, we have $Q(W,y) = \prod_{i \in S} x_i \prod_{i \in T} (1 - x_i)$ and $Q(W,\tilde Y) = \prod_{i \in S} \tilde X_i \prod_{i \in T} (1 - \tilde X_i) $. Now apply Theorem~\ref{final-thm2}.  
\end{proof}

This technique of creating ``dummy elements'' (in this case, the elements $\bar i$), as indicators for not selecting items, will appear in a number of constructions.

\section{Pseudo-approximation algorithm for single-knapsack median}
\label{sec:knapsack-median}
We now describe a $(1+\sqrt{3}+\gamma)$-pseudo-approximation algorithm for single-knapsack median. This is inspired by an approximation algorithm of Li \& Svensson \cite{DBLP:conf/stoc/LiS13}  for $k$-median. The idea is to solve a relaxation called a ``bi-point solution'', and then round it to an additive pseudo-solution. This can also be used to obtain a multiplicative pseudo-approximation. The $k$-median algorithm of \cite{DBLP:conf/stoc/LiS13} has an additional postprocessing step to correct it to a  true solution; however, this step does not seem to work for knapsack median.

Recall that we define the cost for a facility set $\S \subseteq \F$ by $\cost(\S) = \sum_{j \in \C} d(j,\S)$ and the weight by $M(\S) = \sum_{i \in \S} M(i)$. We define $\OPT$ to be the minimum value of $\cost(\S)$ over feasible sets $\S$.  By a straightforward adaptation of \cite{book:ws, jms} to the knapsack setting, we get the following result:
\begin{theorem}
\label{thm:bi_point_km}
  There is a polynomial-time algorithm to compute two sets $\F_1, \F_2 \subseteq \F$ and a parameter $b \in [0,1]$ satisfying the following properties:
  \begin{itemize}
  \item $M(\F_1) \leq 1 \leq M(\F_2)$,
  \item $(1-b) \cdot M(\F_1) + b \cdot M(\F_2) \leq 1$
  \item $(1-b)\cdot \cost(\F_1) + b\cdot \cost(\F_2) \leq 2 \cdot \OPT$.
    \end{itemize}
\end{theorem}

The sets $\mathcal F_1, \mathcal F_2$ are called the \emph{bi-point solution}. For $i \in \F_2$ we define $\sigma(i)$ to be the closest facility of $\F_1$. Following \cite{DBLP:conf/stoc/LiS13}, we define $\text{Star}(i)$ for each facility $i \in \F_1$ to be the set of facilities $k \in \F_2$ with $\sigma(k) = i$, that is, $\text{Star}(i) = \sigma^{-1}(i)$. The intent is that for each $i \in \F_1$, with probability $1-b$ we  open $i$ and with the complementary probability $b$ we open \emph{all} the facilities of $\text{Star}(i)$. In order to preserve the knapsack constraints, there are a few exceptional cases where we open both $i$ and some subset of $\text{Star}(i)$.

The full details of our rounding algorithm are spelled out in Algorithm~\ref{algo:round_star_srdr}. Here, $t$ is an integer parameter to be specified; we assume throughout that $t \geq 20000$.  Note that the modified constraint matrix $M'$ can have negative entries, but this does not cause a problem for executing KPR.

\begin{algorithm}[H]
  \label{algo:roundstars}
\caption{\sc RoundStars($t$)}
\begin{algorithmic}[1]
\STATE Define vectors $y, M' \in [0,1]^{\F_1}$ by setting $y_i = 1 - b, M'(i) = M(i) - M( \text{Star}(i) )$ for each $i \in \F_1$.
\STATE $\tilde Y \gets \textsc{KnapRound}(y, M', t)$
\STATE Define the vector $z \in [0,1]^{\F_2}$, by setting $z_i = 1 - \tilde Y_{\sigma(i)}$.
\STATE  $\tilde Z \gets \textsc{KnapRound}(z, M, t)$
\RETURN $\S = \{i \in \F_1 \mid \tilde Y_i > 0 \} \cup \{ i \in \F_2 \mid \tilde Z_i > 0 \}$.
\end{algorithmic} 
\label{algo:round_star_srdr}
\end{algorithm}

\begin{proposition} 
  \label{claim:weight_bound}
  The solution $\S$ is a $4 t$-additive pseudo-solution with probability one.
\end{proposition}
\begin{proof}
Let $\S'_1$ denote the facilities $i \in \F_1$ with $\tilde Y_i \in (0,1)$ and let $\S'_2$ denote the set of facilities $i \in \F_2$ with $\tilde Z_i \in (0,1)$, and let  $\S' = \S'_1 \cup \S'_2$. Since $\tilde Y$ and  $\tilde Z$ have at most $2 t$ fractional entries, we have $|\S'| \leq 4 t$.  We now claim $M (\S - \S') \leq 1$, which  shows that $\S$ is a $4t$-additive pseudo-solution. We have: 
$$
M( \S - \S') = \sum_{i \in \F_1} [[\tilde Y_i = 1 ]] M(i) + \sum_{i \in \F_2} [[ \tilde Z_i = 1 ]] M(i) \leq \sum_{i \in \F_1} M(i) \tilde Y_i + \sum_{i \in \F_2} M(i) \tilde Z_i.
$$
By Property (E4) of KPR, this equals 
$$
\sum_{i \in \F_1} M(i) \tilde Y_i  +  \sum_{i \in \F_2} M(i) z_i = \sum_{i \in \F_1} M(i) \tilde Y_i  +  \sum_{i \in \F_2} M(i) \cdot (1 - \tilde Y_{\sigma(i)}).
$$
Here,  $\sum_{i \in \F_2} M(i) \cdot (1 - \tilde Y_{\sigma(i)})$ contributes $M( \text{Star}(i) )(1 - \tilde Y_i)$ for each $i \in \F_1$. Thus, the sum is at most
$$
\sum_{i \in \F_1} M(i) \tilde Y_i  + (1 - \tilde Y_i) M( \text{Star}(i))  = \sum_{i \in \F_1}  M( \text{Star}(i)) + M'(i) \tilde Y_i
$$
By Property (E4) of (KPR), this equals $\sum_{i \in \F_1}  M( \text{Star}(i)) + M'(i) y_i$, which can be simplified as 
$$
\sum_{i \in \F_1} M( \text{Star}(i)) + (1-b) (M(i) - M(\text{Star}(i))) = (1-b) M(\F_1) + b M(\F_2)
$$
which is at most $1$ by the properties of the bi-point solution.
\end{proof}

\begin{proposition}
\label{client-bnd-1}
For any facilities $i_1 \in \mathcal F_1, i_2 \in \mathcal F_2$, Algorithm~\ref{algo:round_star_srdr} yields
$$
\Pr(i_1 \notin \S ) \leq b, \qquad \Pr( i_2  \not \in \S) \leq 1-b, \qquad \Pr( i_1 \not \in \S \wedge i_2 \not \in \S ) \leq b (1-b) (1 + O( 1/t))
$$
\end{proposition}
\begin{proof}
First, we have
$$
\Pr( i_1 \not \in \S ) = \Pr( \tilde Y_{i_1} = 0 ) \leq \bE[ 1 - \tilde Y_{i_1} ] =  1 - y_{i_1} = b.
$$

Next, let $k = \sigma(i_2)$. Conditioned on the vector $\tilde Y$ we have 
\begin{equation}
\label{gga1}
\Pr( i_2 \not \in \S \mid \tilde Y ) = \Pr( \tilde Z_{i_2} = 0 \mid z ) \leq \bE[ 1 - \tilde Z_{i_2} \mid z] =  1 - z_{i_2} = \tilde Y_k
\end{equation}

Integrating over $z$ shows that  $\Pr( i_2 \not \in \S ) \leq  \E[\tilde Y_k] = y_k = 1 - b$, as claimed.

For the final result, note that in order to have $i_1 \not \in \S$ and $i_2 \not \in \S$, we must have $\tilde Y_{i_1} = \tilde Z_{i_2} = 0$ and $i_1 \neq k$. So by Eq.~(\ref{gga1}) we have
$$
\Pr( i_1 \not \in \S \wedge i_2 \not \in \S \mid \tilde Y) = [[\tilde Y_{i_1} = 0]] \Pr (i_2 \notin \S \mid \tilde Y) \leq [[\tilde Y_{i_1} =0 ]] \tilde Y_k \leq (1 - \tilde Y_{i_1}) \tilde Y_k.
$$
Since $i_1 \neq k$, and we are assuming $t \geq 20000$,  Proposition~\ref{KPR-dep-round-prop} gives:
\[
 \bE[ \tilde Y_{i_1} ( 1 - \tilde Y_k) ] \leq y_{i_1} (1 - y_k) e^{O(m^2 (1+1)^2/t)}= b (1-b) (1 + O(1/t)) \qedhere
\]
\end{proof}

\begin{proposition}
  \label{client-bnd-2}
  Each client $j \in \C$ has expected cost 
$$
\E[d(j, \S)] \leq (1+ O( 1/t)) \cdot \Bigl( (1-b) d(j, \F_1) + b d(j, \F_2) + 2 d(j, \F_2) b (1-b) \Bigr)
$$
\end{proposition}
\begin{proof}
Let $i_1, i_2$ denote the closest facilities in $\F_1, \F_2$ to $j$, and let $d_1 = d(j, \F_1) = d(j, i_1)$ and $d_2 = d(j, \F_2) = d(j, i_2)$. Clearly if $i_2 \in \S$ then $d(j, \S) \leq d_2$ and likewise if $i_1 \in \S$ then $d(j, \S) \leq d_1$. If neither holds, then necessarily facility $\sigma(i_2)$ is open, in which case we have
  $$
d(j, \S) \leq d(j, \sigma(i_2)) \leq d(j, i_2) + d(i_2, \sigma(i_2)) \leq d(j,i_2) + d(i_2, i_1) \leq d(j,i_2) + d(i_2, j) + d(j, i_1) = 2 d_2 + d_1
$$
 where the inequality $d(i_2, \sigma(i_2)) \leq d(i_2, i_1)$ holds since $\sigma(i_2)$ is the closest facility in $\F_1$.
 
 Now, if $d_1 \geq d_2$, we can calculate 
 $$
 \E[d(j, \S)] \leq d_2 + (d_1 - d_2)\Pr( i_2 \not \in \S ) + \Pr ( i_1 \not \in \S \wedge i_2 \not \in \S ) (2 d_2)
 $$
 and by Proposition~\ref{client-bnd-1}, we have $\Pr( i_2 \not \in \S ) \leq 1-b$ and $\Pr(i_1 \not \in \S \wedge i_2 \not \in \S) \leq b (1 - b) (1 + O( 1/t))$.
 
 Otherwise, if $d_1 < d_2$, we have
 $$
 \E[d(j, \S)] \leq d_1 + (d_2 - d_1)\Pr( i_1 \not \in \S ) + \Pr ( i_1 \not \in \S \wedge i_2 \not \in \S ) (d_1 + d_2)
 $$
 and again by Proposition~\ref{client-bnd-1}, we have $\Pr( i_1 \not \in \S ) \leq b$ and $\Pr(i_1 \not \in \S \wedge i_2 \not \in \S) \leq b (1 - b) (1 + O( 1/t))$.
 \end{proof}

We are now ready to obtain our bi-factor approximation algorithm.
\begin{theorem}
  \label{thm2}
There is an algorithm with $\poly(n/\gamma)$ runtime to obtain an $O(1/\gamma)$-additive pseudo-solution $\S$ with $\cost(\S) \leq (1 + \sqrt{3} + \gamma) \cdot \OPT$.
\end{theorem}
\begin{proof}
We will use Algorithm~\ref{algo:round_star_srdr}, and output either the solution $\S$ it returns or the feasible solution $\F_1$ (whichever has least cost).  By Proposition~\ref{claim:weight_bound}, the solution $\S$ is a $4t$-additive pseudo-solution.

Define $D_1 = \sum_{j} d(j, \F_1)$ and $D_2 = \sum_j d(j, \F_2)$. Applying Proposition~\ref{client-bnd-2} and summing over all clients $j \in \C$, we see that $\S$ satisfies
$$
\E[\cost(\S)] \leq (1 + O(1/t)) \cdot \Bigl( (1-b) D_1 + b D_2 + 2 b ( 1 - b) D_2 \Bigr)
$$

Since the cost of $\A$ is the minimum of the cost of $\F_1$ and $\S$, this implies that
\begin{align*}
	\E[\cost(\A)] \leq (1 + O(1/t)) \cdot \min\{D_1, (1-b) D_1 + b D_2 + 2 b(1-b) D_2 \}
\end{align*}

It can be routinely verified  that\footnote{For example, this inequality can be encoded in the first-order theory of real-closed fields,  in terms of indeterminates $b, D_1, D_2$. This theory is decidable, so it can be checked that it holds for all values $b \in [0,1], D_1 \geq 0, D_2 \geq 0$.}
$$
\min \{ D_1, (1-b) D_1 + b D_2 + 2 b(1-b) D_2 \} \leq \frac{1 + \sqrt{3}}{2} \bigl ( (1-b) D_1 + b D_2 \bigr)
$$

Since $\mathcal F_1, \mathcal F_2$ is a bi-point solution, we have $(1 - b) D_1 + b D_2 \leq 2 \cdot  \OPT$. Therefore,
$$
\E[\cost(\A)] \leq (1 + O(1/t)) \cdot 2 \cdot \OPT \cdot \frac{1+\sqrt{3}}{2} = (1+O(1/t)) \cdot (1+\sqrt{3}) \cdot \OPT
$$

If we set $t = \Omega(1/\gamma)$, then after an expected $O(1/\gamma)$ repetitions of this process, we obtain a solution of cost  at most $(1+O(\gamma)) \cdot \OPT$. 
\end{proof}

We can leverage this to obtain a multiplicative pseudo-approximation.
\begin{theorem}
  There is an algorithm with $n^{O(\epsilon^{-1} \gamma^{-1})}$ runtime to obtain an $\epsilon$-multiplicative pseudo-solution $\S$ with $\cost(\S) \leq (1 + \sqrt{3} + \gamma) \cdot \OPT$.
\end{theorem}
\begin{proof}
Call a facility $i$ \emph{big} if $M(i) \geq \rho = \Theta( \epsilon \gamma )$. The solution may have at most $1/\rho$ big facilities, which we can guess in $n^{O(1/\rho)}$ time. Now construct a residual instance where all other big facilities, aside from the ones guessed to be in our solution, are removed. Apply Theorem~\ref{thm2} to this residual instance; after rescaling, the resulting solution $\S$ satisfies $\cost(\S) \leq (1 + \sqrt{3} + \gamma) \cdot \OPT$ and $M(\S) \leq 1 + O(\rho/\gamma)$,  since each facility now has weight at most $\rho$. Our choice of $\rho$ ensures that $M(\S) \leq 1 + \epsilon$ and gives a runtime of $n^{O(1/\rho)} \cdot (n/\gamma)^{O(1)} n^{O(\epsilon^{-1} \gamma^{-1})}$. 
\end{proof}

\section{Pseudo-approximation algorithm for multi-knapsack median}
\label{sec:knapsack-median2}
In this section, we give a $3.25$-pseudo-approximation algorithm for multi-knapsack median. This is based on applying KPR for a key rounding step in the $3.25$-approximation algorithm of Charikar \& Li \cite{charikar_dependent} for $k$-median. Although this is not the best approximation ratio for $k$-median, the main benefits of that algorithm is its good approximation ratio as a function of the ``obvious'' LP relaxation defined as follows:
\begin{alignat*}{3}
\text{ minimize }   &  \sum_{j \in \C} \sum_{i \in \F} x_{i,j} d(i,j) \\
    \text{subject to } &  \sum_{i \in \F} x_{i,j} = 1 \qquad \forall j \in \C \\
     &  0 \leq x_{i,j} \leq y_i \leq 1 \qquad \forall i  \in \C, j \in \F \\
     &  M y \leq \vec 1  
  \end{alignat*}

Here, $x_{i,j}$ represents fractionally how client $j$ is matched to facility $i$, and $y_i$ is an indicator that facility $i$ is open. The final constraint specifies that each of the $m$ knapsack constraints is satisfied. For any client $j \in \C$, let $r_j = \sum_i x_{ij} d(i,j)$ denote the fractional connection cost of $j$. By standard facility-splitting methods (see e.g., \cite{swamy2004approximation}), we can ensure that $y_i > 0$ and $x_{i,j} \in \{0, y_i \}$ for all $i,j$. 

\subsection{The Charikar-Li algorithm}
Given an LP solution,  the Charikar-Li algorithm has two phases, which we briefly summarize here. Please see \cite{charikar_dependent} for further details. The \emph{bundling} phase can be divided into four stages:
\begin{enumerate}
\item A client set $\C' \subseteq \C$ is chosen, such that the clients $j, j' \in \C'$ are relatively far apart from each other. For $j \in \C$ let $\sigma(j)$ be the closest client in $\C'$.
\item For each $j \in \C'$, we define a set $\mathcal U_j \subseteq \F$ which are the facilities ``claimed by'' $j$. The sets $\mathcal U_j$ are called ``bundles''; they are disjoint and have $1/2 \leq y(\mathcal U_j) \leq 1$. We define $R_j = \tfrac{1}{2} d(j, \C' - j)$ for $j \in \C'$.
\item The clients in $\C'$ are paired up, giving a partition of $\C'$ into cardinality-two sets.\footnote{If $|\C'|$ is odd, then $\mathcal C'$ has one remaining unmatched client. The Charikar \& Li algorithm has some additional steps to handle this case. We can avoid these exceptional steps by adding an additional dummy client and dummy zero-cost facility, with distance zero to each other and distance $\infty$ to all other clients. This allows us to assume without loss of generality that $|\C'|$ is even.} We refer to this as the \emph{matching $\mathcal M$ of $\mathcal C'$}.
\item We also define $\mathcal U_0 = \mathcal F - \bigcup_{j \in \mathcal C'} \mathcal U_j$; these are the ``unbundled'' facilities.
\end{enumerate}

In the \emph{selection} phase, each pair $\{ j, j' \} \in \M'$ selects either one or two facilities to open from $\mathcal U_j \cup \mathcal U_{j'}$, and in addition some unbundled facilities are opened. 

The simplest strategy for this, which is  a baseline for more advanced algorithms, is  \emph{independent selection}. Here, we   first choose a set of ``open clients'' in $\C'$, wherein for each pair $\{ j, j' \} \in \M$, we open $j$ or $j'$ with the following probabilities: (i) open $j$ alone with probability $1 - y(\mathcal U_{j'})$, (ii) open  $j'$ alone with probability $1 - y(\mathcal U_{j})$, or (iii) open both $j, j'$ with probability $y(\mathcal U_j) + y(\mathcal U_{j'}) - 1$. Then, for each open client $j$, we open exactly one facility $i$ in $\mathcal U_j$, wherein $i$ is chosen with probability proportional to $y_i$.   Also, for each facility $i \in \mathcal U_0$, we open $i$ independently with probability $y_i$.

Charikar \& Li shows a number of powerful bounds and properties for this process.
\begin{proposition}[\cite{charikar_dependent}]
\label{neg-cor}
Under independent selection, the following bounds hold:
\begin{itemize}
\item[(B1)] $\Pr( i \in \S) \leq y_i$ for any facility $i$.
\item[(B2)] The events $[[i \in \S]]$, for $i \in \F$, are cylindrically negatively correlated.
\item[(B3)] For any client $j \in \C$ we have $\E[d(j, \S)] \leq 3.25 r_j$.
\item[(B4)] For any  client $j \in \C$ we have $r_{\sigma(j)} \leq r_j$ and $d(j, \sigma(j)) \leq 4 r_j$
\item[(B5)] For any client $j \in \C'$, it holds that $d(j, \S) \leq  \beta R_{j}$ with probability one, where $\beta$ is some constant.
\item[(B6)] For any client $j \in \C'$, there holds $r_{j} \leq R_j$ and $\mathcal B(j, R_j) \subseteq \mathcal U_j$
\end{itemize}
\end{proposition}

This leads to a simple multiplicative pseudo-approximation algorithm for multi-knapsack median.
\begin{theorem}
  \label{jkj1}
  Let $\gamma, \epsilon \in (0,1)$. There is an algorithm with $n^{O(m \log(m/\gamma)/\epsilon^2)}$ runtime to obtain an $\epsilon$-multiplicative pseudo-solution $\S$ with $\cost(\S) \leq (3.25 + \gamma) \cdot \OPT$.
\end{theorem}
\begin{proof}[Proof (Sketch)]
Say a facility $i \in \F$ is big if $M_{k}(i) \geq \rho = \frac{\epsilon^2}{10 \log (m/\gamma)}$ for any $k \in [m]$. We can guess the big facilities in an optimal solution in $n^{O(m/\rho)}$ time. Next, solve the LP and run the Charikar-Li rounding on the resulting residual instance.  By properties (B1), (B2), each $\sum_i M_{k}(i) [[i \in \S]]$ is a sum of  negatively-correlated random variables with mean $\sum_i M_{k}(i) y_i \leq 1$, and each is bounded in the range $[0,\rho]$. By Chernoff's bound, the probability that it exceeds $1+\epsilon$ is at most $e^{-\epsilon^2/(3 \rho)} = O(\gamma/m)$.
\iffalse
To simplify the notation, rescale so that $\OPT = 1$. We say a facility $i \in \F$ is big if $M_{k}(i) \geq \rho$ for any $k \in [m]$. We begin by guessing the set of opened big facilities; this takes time $n^{O(m/\rho)}$.

Next, we solve the LP and run the Charikar-Li rounding on the resulting residual instance. Letting $V = \sum_{i \in \C} d(i, \S)$, Proposition~\ref{neg-cor} shows that $\E[V] \leq 3.25$. Because we have guessed the big facilities, the remaining facilities have $M_{k}(i) \leq \rho$. By Proposition~\ref{neg-cor}, each term $\sum_i M_{k}(i) Y_i$ is a sum of negatively-correlated random variables with mean at most $\sum_i M_{k}(i) y_i \leq 1$. By Chernoff's bound, the probability it exceeds $1+\epsilon$ is at most $e^{-\epsilon^2/(3 \rho)}$. Thus, the overall probability that weight constraint is violated, is at most $m e^{-\epsilon^2/(3 \rho)}$.

Let $\mathcal E$ denote the event that $M Y \leq (1 + \epsilon) \vec 1$. With $\rho = \frac{\epsilon^2}{10 \log (m/\gamma)}$, we see that $\Pr(\mathcal E) \geq 1 - \gamma / 3 \geq 1/2$. Thus, $\bE[ V \mid \mathcal E ] \leq (3.25 + O(\gamma))$, and we can sample quickly from the conditional distribution on $\mathcal E$.  Furthermore, we have $\Pr(V \leq 3.25 + 2 \gamma \mid \mathcal E) \geq \Omega(\gamma)$, which implies that after an expected $O(1/\gamma)$ repetitions we find a set $\S$ such that $V \leq 3.25 + 2 \gamma$ and such that $\mathcal E$ occurs.

The overall runtime is $O(1/\gamma) \cdot n^{O(1)} \cdot n^{O(m/\rho)}$; the result follows by rescaling and simplifying.
\fi
\end{proof}

\subsection{Facility selection as a knapsack-partition system}
In the selection phase, each pair $e \in \mathcal M$ needs to select one or two facilities, according to a certain probability distribution. We also need to decide the status of the unbundled facilities. We can encode these requirements in terms of a knapsack-partition system. Here, the ground-set $U$ will be a polynomial-size subset of the power set $2^{\overline{\mathcal F}}$, where $\overline{\mathcal F}$ is $\mathcal F$ plus some dummy items.  That is, each element of $U$ is itself a set containing (possibly dummy) facilities.

For this construction, for each pair $e = \{ j, j' \} \in \M$, we create a block $G_e \subseteq 2^{\F}$ defined as
$$
G_e =  \Bigl \{   \{ i \} \mid i \in \mathcal U_j \cup \mathcal U_{j'} \Bigr \} \cup \Bigl \{ \{i, i' \},  \mid i \in \mathcal U_j, i' \in \mathcal U_{j'} \Bigr \}
$$

For each unbundled facility $i \in \mathcal U_0$, we create a block $G_i \subseteq 2^{\overline \F}$ defined as
$$
G_i = \bigl \{ \{ i \}, \{ \bar i \} \bigr \}
$$
where $\bar i$ is a dummy item (which indicates that facility $i$ is not to be opened). 

The ground set is $U = \bigcup_{e \in \mathcal M} G_e \cup \bigcup_{i \in \mathcal U_0} G_i$. Since the sets $\mathcal U_j$ are disjoint, the blocks $G_e$ and $G_i$ form a partition of $U$, which we denote by $\mathcal G$.   Correspondingly, we also define a vector $z \in [0,1]^{U}$ as follows. For each $e = \{ j, j' \} \in \M$, and every pair of facilities $i \in \mathcal U_j, i' \in \mathcal U_{j'}$, we set:
\begin{align*}
&z_{ \{i \} } = (1 - y(\mathcal U_{j'})) \frac{y_i}{y(\mathcal U_j)}, \qquad \qquad z_{ \{i' \} } = (1 - y(\mathcal U_{j})) \frac{y_{i'}}{y(\mathcal U_{j'})} \\
&\qquad \qquad z_{ \{i,i' \} } = (y(\mathcal U_j) + y(\mathcal U_{j'}) - 1) \frac{y_i y_{i'}}{y(\mathcal U_j) y(\mathcal U_{j'})} 
\end{align*}

Likewise, for each $i \in \mathcal U_0$, we define 
$$
z_{ \{i \} } = y_i, \qquad z_{ \{ \bar i \} } = 1 - y_i
$$

Finally, we extend the knapsack constraints to $U$, by setting $M_k(\bar i) = 0$ for each dummy item $i$ and setting $M_k( V) = \sum_{i \in V} M_k(i)$ for $V \subseteq 2^{\overline{\F}}$. It can be seen that that $z$ is a fractional solution to the partition system; also, by the way we have extended $M$ to $U$, we have $M z = M y$.

Given any facility set $W \subseteq \F$, we define a corresponding set $W^* \subseteq U$ by $W^* = \{ V \in U \mid V \cap W \neq \emptyset \}$. The following observation summarizes how the knapsack partition system is connected to the Charikar-Li selection phase and independent selection. 
\begin{observation}
  \label{gt110}
  Given an integral vector $Z \in \{0, 1 \}^U$, we can generate a corresponding facility set $\S$ by opening all (non-dummy) facilities in all sets $V$ with $Z_V = 1$, i.e., $\S = \bigcup_{V: Z_V = 1} V \cap \mathcal F$. In this case, we have the following:
  \begin{itemize}
  \item There holds $M( \S) \leq M Z$.
    \item For any set of facilities $W \subseteq \F$, there holds $\S \cap W \neq \emptyset$ if and only $Q(W^*, Z) = 0$.
      \item If $Z$ is generated as $Z = \textsc{IndSelect}(\mathcal G, z)$, then the resulting solution set $\S$ has the same probability distribution as in independent selection for the Charikar-Li algorithm.
  \end{itemize}
\end{observation}

We omit the proofs since they follow immediately from definitions. As one important consequence of Observation~\ref{gt110}, we can write the expected distance for a client $j$ in terms of the potential function $Q$ for the resulting knapsack-partition system:
\begin{proposition}
  \label{gt13}
If $\S$ is generated as in Observation~\ref{gt110}, then any client $j$ has 
$$
d(j, \S) = \int_{u=0}^{4 r_j + \beta R_{\sigma(j)}} Q(\mathcal B(j,u)^*, Z) \ du.
$$
\end{proposition}
\begin{proof}
 Let $s = 4 r_j + \beta R_{\sigma(j)}$. By properties (B4) and (B5), we have $d(j, \S) \leq d(j, \sigma(j)) + d(\sigma(j), \S) \leq s$ with probability one. For any $u \geq 0$, we have $\S \cap \mathcal B(j,u) = \emptyset$ if and only if $ Q( \mathcal B(j,u)^*, Z) = 1$, and so
 \[
    d(j, \S) = \int_{u=0}^{s} [[d(j, \S) > u]] \ du = \int_{u=0}^{s} [[\S \cap \mathcal B(j,u) = \emptyset]] \ du = \int_{u=0}^{s} Q( \mathcal B(j,u)^*, Z) \ du \qedhere
\]
  \end{proof}

\subsection{KPR selection strategy}
Our approximation algorithm, summarized as Algorithm~\ref{algo:deproundmknap}, uses \textsc{FullKPR} rounding instead of independent selection in the Charikar-Li algorithm.
\begin{algorithm}[H]
\caption{$\textsc{MultiKnapsackMedianRound} (t)$}
\label{algo:deproundmknap}
\begin{algorithmic}[1]
\STATE Let $x,y$ be the solution to the LP.
\STATE Run the Charikar-Li bundling phase, resulting in fractional vector $z \in [0,1]^{U}$ and partition $\mathcal G$.
\STATE Let $Z = \textsc{FullKPR}(\mathcal G, M, z, t)$.
\RETURN $\S = \bigcup_{V: Z_V = 1} V \cap \mathcal F$.
\end{algorithmic} 
\end{algorithm}

Our main rounding result is that Algorithm~\ref{algo:deproundmknap} has a similar probability of opening a facility in any given set $W$ compared to independent selection. 
\begin{lemma}
\label{knap-mainl1}
  Let $t \geq 20000 m^2$. Then for any client $j \in \C'$ and any set $W \subseteq \F$, Algorithm~\ref{algo:deproundmknap} satisfies
  $$
 \E[ Q( W^*,  Z) ] \leq Q(W^*, z) + O(m^2/t)\cdot \Bigl( \frac{r_j}{R_j} + \sum_{i \in \mathcal U_j - W} x_{i,j} \Bigr)
  $$
\end{lemma}
\begin{proof}
Let $e = \{j, j' \} \in \mathcal M$ be the pair in the matching corresponding to $j$. We apply Theorem~\ref{final-thm2}(a) with respect to 
  $\mathcal D = \{G_{e} \}$. Since $|\mathcal D| = 1$, this gives 
  $$
  \E[Q( W^*, Z)] \leq Q(W^*, z) + O( m^2 / t) \cdot Q( W^* \wedge \mathcal D, z) 
  $$
  
  To finish the proof, we need to show that 
  \begin{equation}
  \label{tty10}
Q(W^* \wedge \mathcal D, z) \leq r_j/R_j + \sum_{i \in \mathcal U_j - W} x_{i,j}.
\end{equation}

For this, we calculate:
 \begin{align*}
  Q(W^* \wedge \mathcal D, z) &= 1 - \sum_{ \substack{ V \cap W \neq \emptyset \\ V \in G_{e} } }z_V \leq 1 - \sum_{ \substack{i \in W \cap \mathcal U_j }} z_{ \{ i \} }  - \sum_{i \in W \cap \mathcal U_j, i' \in \mathcal U_{j'}} z_{ \{i, i' \} } \\
&  = 1 - \sum_{i \in \mathcal U_j \cap W}\frac{ (1 - y(\mathcal U_{j'}))  y_i}{y(\mathcal U_j)} - \sum_{\substack{i \in W \cap \mathcal U_j,  i' \in \mathcal U_{j'}}} \frac{ (y(\mathcal U_j) + y(\mathcal U_{j'}) - 1)  y_i y_{i'} }{ y(\mathcal U_j) y(\mathcal U_{j'}) } \\
    &= 1 - y(\mathcal U_j \cap W) 
  \end{align*}
  
The facility-splitting step ensures that $y_i = x_{i,j}$ for any $i \in \mathcal U_j$. So $
  1 - y( W \cap \mathcal U_{j}) = 1 - \sum_{i \in W \cap \mathcal U_{j}} x_{i,j} = 1 - \sum_{i \in \mathcal U_j} x_{i,j} +  \sum_{i \in \mathcal U_j - W} x_{i,j}$. Finally, to bound the term $1-\sum_{i \in \mathcal U_j} x_{i,j}$, we use property (B6):
\[
    1-\sum_{i \in \mathcal U_j} x_{i,j} \leq 1-\sum_{i \in \mathcal B(j, R_j)} x_{i,j} = \sum_{i: d(i,j) > R_j} x_{i,j} \leq \sum_i \frac{d(i,j)}{R_{j}} x_{i,j} = \frac{r_{j}}{R_{j}}
    \]
    
This shows Eq.~(\ref{tty10}) and hence shows the claimed result.
\end{proof}

This, in turn, allows us to bound expected distances for the rounding algorithm.
\begin{proposition}
  If $t \geq 20000 m^2$, then for any client $j$ the set $\S$ returned by Algorithm~\ref{algo:deproundmknap} satisfies $\E[d(j, \S)] \leq D_j + O( r_j m^2 / t)$, where $D_j$ is the expected distance of $j$ under independent selection.
  \end{proposition}
\begin{proof}
For $u \geq 0$, define $C_u = \mathcal B(j,u)^*$, and let $j' = \sigma(j)$. By Proposition~\ref{gt13}, the expected distance of $j$ is given by:
\begin{align*}
  \E[d(j, \S)] &= \int_{u=0}^{s} \bE[Q( C_u,Z )] \ du
\end{align*}
where $s = 4 r_j + \beta R_{j'}$. Using Lemma~\ref{knap-mainl1} with respect to client $j'$ and facility set $W = \mathcal B(j, u)$ gives:
\begin{align*}
  \int_{u=0}^s \bE[Q( C_u, Z)] \ du &\leq \int_{u=0}^s \Biggl( Q( C_u, z) + O(m^2/t) \Bigl( \frac{r_{j'}}{R_{j'}} + \sum_{i \in \mathcal U_{j'} - \mathcal B(j,u)} x_{i,j'} \Bigr) \Biggr) \ du \\
  &= \int_{u=0}^s Q(C_u, z) \ du + O(m^2/t) \Biggl( \frac{s r_{j'}}{R_{j'}} + \int_{u=0}^s  \sum_{i \in \mathcal U_{j'} - \mathcal B(j,u)} x_{i,j'} \ du \Biggr)
\end{align*}

Now note that, under independent selection, we have $\bE[ Q(C_u, Z) ] = Q(C_u, z)$ for each $u$. Thus, by Proposition~\ref{gt13}, the first term here is precisely $D_j$.  To estimate the last term, we interchange summation and use the triangle inequality to get:
\begin{align*}
& \int_{u=0}^s \sum_{ i \in \mathcal U_{j'} - \mathcal B(j,u) } x_{i,j'} \ du = \sum_{i \in \mathcal U_{j'}} x_{i, j'} \int_{u=0}^s [[d(i,j) > u]] \ du = \sum_{i \in \mathcal U_{j'}} x_{i, j'} \min(s, d(i,j)) \\
  &\qquad \leq \sum_{i \in \mathcal U_{j'}} x_{i, j'} \bigl( d(j,j') + d(i, j') \bigr)  = d(j, j') \sum_{i \in \mathcal U_{j'}} x_{i, j'} + \sum_{i \in \mathcal U_{j'}} d(i,j') x_{i,j'} \leq d(j, j') + r_{j'}  
\end{align*}

  Therefore
$$
  \E[d(j, \S)] \leq D_j + O(m^2/t) \Bigl( \frac{ s r_{j'}}{R_{j'}} + d(j,j') + r_{j'} \Bigr)
$$

  Now, since $r_{j'} \leq r_j$ and $r_{j'} \leq R_{j'}$, we have  $s r_{j'}/R_{j'} = 4 r_j (r_{j'}/R_{j'}) + r_j' \cdot \beta R_{j'} / R_{j'} \leq O(r_j)$. Also, by (B4), we have $d(j, j') + r_{j'} \leq O(r_j)$. Overall,  we see that
  $$
  \frac{ s r_{j'}}{R_{j'}} + d(j,j') + r_{j'}  \leq O( r_j ),
  $$
  which concludes the proof.
\end{proof}

From property (B3), this immediately shows the following:
\begin{corollary}
\label{cor-aratio}
If $t \geq 20000 m^2$, then Algorithm~\ref{algo:deproundmknap} ensures $\E[d(j, \S)] \leq (3.25 + O(m^2/t)) r_j$ for every client $j \in \C$, and furthermore $\E[ \cost(\S) ] \leq (3.25 + O(m^2/t)) \cdot \OPT$.
\end{corollary}

We now show Theorem~\ref{thm:multi-knapsack-medianx}, restated here for convenience.
{
\renewcommand{\thetheorem}{\ref{thm:multi-knapsack-medianx}}
\begin{theorem}

  Consider a multi-knapsack median instance with $m$ constraints, and let $\epsilon, \gamma \in (0,1)$. There is an algorithm with $\poly(n/\gamma)$ runtime to obtain an $O(\tfrac{m}{\sqrt{\gamma}})$-additive pseudo-solution $\S$ with $\cost(\S) \leq (3.25 + \gamma) \cdot \OPT$, and an  algorithm with $n^{\tilde O( m^2 \epsilon^{-1} \gamma^{-1/2})}$ runtime to obtain an $\epsilon$-multiplicative pseudo-solution $\S$ with $\cost(\S) \leq (3.25 + \gamma) \cdot \OPT$.
\end{theorem}
}
\begin{proof}
 Let  us define $V = \cost(\S) / \text{OPT}$; Proposition~\ref{neg-cor} shows that $\E[V] \leq (3.25 + O(m^2/t))$ for $t \geq 20000 m^2$. We will choose $t = m^2/\gamma$ (which satisfies $t \geq 20000 m^2$ for $\gamma$ sufficiently small).

Let $\mathcal E$ be the event that $Z$ is an $q$-additive pseudo-solution for the given value $q$. By Theorem~\ref{full-KPR-thm1}, we have $\Pr( \mathcal E) \geq 1 - \gamma$.  Since each element in $U$ contains at most two facilities, the solution $\S$ is a $2 q$-additive pseudo-solution to the original knapsack whenever $\mathcal E$ holds.  Also, we have $\bE[ V \mid \mathcal E ] \leq 3.25 + O(\gamma)$. So, after an expected $O(1/\gamma)$ repetitions we find a set $\S$ such that $\mathcal E$ holds and such that $V \leq 3.25 + O(\gamma)$.  Overall, we get a runtime of $O(1/\gamma) \cdot n^{O(1)}$; the result follows by rescaling $\gamma$ and simplifying.

For the second result, say a facility $i$ is big if $M_{k}(i) > \rho = \Theta( \frac{\epsilon \sqrt{\gamma}}{m \sqrt{\log(m/\gamma)}}) $ for any $k \in [m]$. We can guess the big facilities in an optimal solution in $n^{O(m/\rho)}$ time. We then construct a residual instance where all other big facilities are discarded, and apply the additive pseudo-approximation to it. The overall runtime is $(n/\gamma)^{O(1)} \cdot n^{O(m/\rho)}$; with simplification of parameters, this gives the claimed runtime.
\end{proof}

\section{The knapsack center problem}
\label{sec:knapsack-center}
We now analyze  the knapsack center problems, proving Theorems~\ref{res:standardKC} and \ref{res:MKC}. There are only $\binom{n}{2}$ possible values for the optimal radius $R = \OPT$, and so we can guess this value in $O(n^2)$ time. To simplify the notation for this section, let us suppose that we have guessed $R$ and rescaled to have $R = 1$. 

   Given some arbitrary $\gamma > 0$,  our goal is to find a distribution $\Omega$ over solution sets $\S$, such that every client $j \in \C$ has  
\begin{equation}
\label{mwueqn1}
{\textstyle \E_{\S \sim \Omega}} [d(j, \S)] \leq 1 + 2/e + \gamma, \qquad \qquad d(j, \S) \leq 3 \text{ with probability one}
\end{equation}
We refer to the distribution $\Omega$ as a \emph{$\gamma$-fair solution}. Let us define a \emph{weighting function} $a$ to be a map $a: \mathcal \C \rightarrow [ 0,1]$ with $\sum_{j \in \mathcal C} a_j = 1$.  In such a distribution $\Omega$, every weighting function $a$ would have a corresponding solution set $\S$ satisfying
\begin{equation}
\label{mwueqn}
\sum_{j \in \C} a_j d(j, \S) \leq 1 + 2/e + \gamma, \qquad \qquad  \max_{j \in \C} j \ d(j, \S) \leq 3
\end{equation}

By LP duality, the converse also holds. Furthermore, the multiplicative weights update (MWU) method makes this efficient: if we have an  efficient algorithm $\mathcal A$ which takes as input a weighting function $a$ and returns a solution set $\S$ satisfying Eq.~(\ref{mwueqn}), then it can be converted into an efficient randomized algorithm which returns a $O(\gamma)$-fair solution. We summarize this as the following Algorithm~\ref{algo:mwu}:
\begin{algorithm}[H]
  \caption{$\textsc{KnapsackCenterMWU}$}
\begin{algorithmic}[1]
\STATE Initialize the vector $a^{(1)}_j = 1$ for all clients $j \in \C$.
\FOR{$k = 1, \dots, v$}
\STATE Use algorithm $\mathcal A$ with the weighting function $\frac{a^{(k)}}{a^{(k)}(\mathcal C)}$ to get solution set $\S_k$.
\STATE\textbf{for} each $j \in \C$ \textbf{do} update $a^{(k+1)}_j = e^{\epsilon d(j, \S_k)} a^{(k)}_j$
\ENDFOR
\STATE Return $\S_K$ for $K$ uniformly chosen from $[v]$.
\end{algorithmic} 
\label{algo:mwu}
\end{algorithm}

\begin{lemma}
\label{glem1}
Suppose that, for every weighting function $a$, the algorithm $\mathcal A$ runs in time $T$ and generates a solution $\S$ satisfying Eq.~(\ref{mwueqn}) from some family $\mathfrak S$. Then with appropriate parameters $v, \epsilon$,  Algorithm~\ref{algo:mwu} runs in $\poly(n/\gamma) \cdot T$ time and outputs an $O(\gamma)$-fair solution $\S$ in $\mathfrak S$.
\end{lemma}
\begin{proof}
Let $\beta = 1 + 2/e + \gamma$. The output $\S_K$ is clearly in $\mathfrak S$, since $\S_1, \dots, \S_k$ are all in $\mathfrak S$. Let us define $\Phi_k = \sum_{j \in \C} a^{(k)}_j$.  Since $d(j, \S_k) \leq 3$, for sufficiently small $\epsilon$ we have
  \begin{align*}
  \Phi_{k+1} &= \sum_{j \in \C} a^{(k)}_j e^{\epsilon d(j, \S_k)} \leq \sum_{j \in \C} a^{(k)}_j (1 + \epsilon d(j, \S_k) + \epsilon^2 d(j, \S_k)^2) \leq a^{(k)}(\mathcal C) + (\epsilon + 3 \epsilon^2) \beta \sum_j a^{(k)}_j d(j, \S)
  \end{align*}
  The algorithm $\mathcal A$ ensures that $\sum_j \frac{a_j^{(k)}}{a^{(k)}(\mathcal C)} d(j, \S)  \leq \beta$ and therefore $$
  \Phi_{k+1} \leq a^{(k)}(\mathcal C) + (\epsilon + 3 \epsilon^2) \beta a^{(k)}(\mathcal C) = (1 + (\epsilon + 3 \epsilon^2) \beta) \Phi_k$$
Since $\Phi_1 \leq n$, this implies that  $\Phi_{v+1} \leq (1 + (\epsilon + 3 \epsilon^2)  \beta)^v n \leq e^{v (\epsilon + 3 \epsilon^2) \beta} n$.
  
  Now consider some client $j$ at the end of this process. We have $e^{\epsilon \sum_{k=1}^v d(j, \S_k)} = a^{(v+1)}_j \leq \Phi_{v+1} \leq e^{v (\epsilon + 3 \epsilon^2) (1 + 2/e + \gamma)} n$. Taking logarithms and simplifying, this shows   $ \sum_{k=1}^v d(j, \S_k) \leq v \beta  (1 + 3 \epsilon) + \log n$.    Therefore,  with $\epsilon = \gamma$ and $v = \frac{\log n}{\gamma^2}$,  the output $S_K$ has $  \E[d(j, \S_K)] = \tfrac{1}{v} \sum_{j=1}^k d(j, \S_k) \leq \beta + O(\gamma).$
  \end{proof}
  
  \subsection{The knapsack center LP}
  We use the following LP consisting of points $(x,y)$ satisfying constraints (C1) -- (C4):
\begin{enumerate}
	\item[(C1)] $\sum_{i \in \mathcal B(j, 1)} x_{ij} = 1$ for all $j \in \C$ (all clients should get connected to some open facility),
	\item[(C2)] $x_{ij} \leq y_i$ for all $i,j \in \C$ (client $j$ can only connect to facility  $i$ if it is open),
	\item[(C3)] $My \leq \vec 1$ (the $m$ knapsack constraints),
	\item[(C4)] $0 \leq x_{ij}, y_i \leq 1$ for all $i,j \in \C$.
\end{enumerate} 

By splitting facilities we may enforce an additional constraint:
\begin{enumerate}
\item[(C5)] For all $i \in \F, j \in \C$, we have $x_{ij} \in \{0, y_i \}$,
\end{enumerate} 

We say a facility $i$ is \emph{integral} if $y_i \in \{0, 1 \}$; else it is \emph{fractional}. For $j \in \C$ we define $F_j := \{i \in \F: x_{ij} > 0\}$ and similarly for $i \in \F$ we define $H_i = \{j \in \C: x_{ij} > 0 \}$. We may form a subset $\C' \subseteq \C$, such that the sets $F_{j}, F_{j'}$ for $j, j' \in \C'$ are pairwise disjoint, and such that $\C'$ is maximal with this property. We also define  $F_0 = \F - \bigcup_{j \in \C'} F_j$.

We turn this into a knapsack-partition instance as follows. For each $j \in \C'$, we define a block $G_j$ to be simply $F_j$. For each $i \in F_0$, we create a dummy item $\bar i$ with $y_{\bar i} = 1 - y_i$ and $M (\bar i) = 0$, and we create a block $G_i = \{ i, \bar i \}$. (If we select $\bar i$, it simply means that we do not choose to open facility $i$). By Property (C5), we have $y(F_j) = 1$ for all $j$, and so $y$ satisfies the partition constraints. 

The following result shows how the knapsack center instance relates to this knapsack-partition system, in particular to the potential function $Q$.
\begin{proposition}
  \label{bsl-guarantee-knap}
  Let $j \in \mathcal C$ be any client.   If $(x,y)$ is a fractional LP solution, then $Q(F_j, y) \leq 1/e$.  If $Y$ is an integral solution to the partition system, and we open the (non-dummy) facilities in the support of $Y$, then $d(j, \S) \leq 1 + 2 Q(F_j, Y)$.
\end{proposition}
\begin{proof}
For the first claim, since $y(F_j) = 1$ and the blocks $G$ are pairwise disjoint, we have 
$$
  Q(F_j,y) = \prod_{G \in \mathcal G} (1 - y(G \cap F_j)) \leq \prod_{G \in \mathcal G} e^{-y(G \cap F_j)} = e^{-y(F_j)} = 1/e
  $$

For the second claim, if some $i \in F_j$ is opened (i.e., if $Q(F_j, Y) = 0$), then $d(j, \S)  \leq 1$. Also,  by maximality of $\C'$, there must exist some facility $k \in \C'$ with $F_j \cap F_k \neq \emptyset$ (possibly $k = j$). There will be some facility opened in $F_k$, and so $d(j, \S) \leq d(j, i) + d(i,k) + d(k,\S) \leq 3$ with probability one.
  \end{proof}
  
\subsection{Removing dense facilities}
Given some  fixed weighting function $a$, we first need a preprocessing step to ensure that no facility $i$ serves a large (weighted) fraction of the clients. 
 
\begin{proposition} 
\label{prune-prop}
For any $\delta > 0$ and weighting function $a$, there is an algorithm with $n^{O(1/\delta)}$ runtime, which returns a fractional LP solution $(x,y)$ such that every fractional facility $i \in \F$ has $a( H_i ) \leq \delta$. A solution $(x,y)$ with this property is  called \emph{$\delta$-sparse} with respect to $a$.
\end{proposition}
\begin{proof}
We recursively execute the following procedure: First, solve the LP to obtain a fractional solution $(x, y)$. Next, if this solution contains some fractional facility $i \in \F$ with $a(H_i) > \delta$, then we form two subproblems; in the first, we force $y_i = 0$ and in the latter, we force $y_i = 1$ and $x_{ij} = 1$ for $j \in H_i$.

Since there is an optimal integral solution, this branching process generates at least one feasible subproblem. Furthermore, each time we execute a branch with $y_i = 1$, the resulting subproblem has $a(\mathcal C_{\text{frac}})$ reduced by at least $\delta$, where $\mathcal C_{\text{frac}}$ denotes the set of clients which are served by a fractional facility. So the search tree has depth at most $1/\delta$. At a leaf of this branching process,  there holds $a(H_i) \leq \delta$ for every fractional facility $i$. Since each subproblem can be solved in $\poly(n)$ time, the overall runtime is $n^{O(1/\delta)}$.
\end{proof}

The next results show how $\delta$-sparse fractional solutions are in certain senses ``stable'' under small modifications, and how this interacts with the KPR rounding process.

\begin{proposition}
\label{mod-prop}
Suppose that $(x,y)$ is a $\delta$-sparse LP solution with respect to weighting function $a$, and vector $y'$ is obtained by modifying $t$ fractional entries of $y$. Then $
\sum_{j \in \C} a_j Q(F_j, y') \leq t \delta + \sum_{j \in \C} a_j Q(F_j, y).
$
\end{proposition}
\begin{proof}
Let $A$ denote the modified facilities and let $V = \bigcup_{i \in A} H_i$; these are the clients which are affected by the modified facilities. For $j \in \C - V$, we have $Q(F_j, y') = Q(F_j, y)$, as none of the clients $i \in F_j$ get modified.  Thus, we can write:
$$
  \sum_{j \in \C} a_j Q(F_j, y') \leq \sum_{j \in \C} a_j Q(F_j, y) + \sum_{j \in V}  a_j Q(F_j, y')
 $$
  
   For the latter term,  $\delta$-sparsity implies $a(H_i) \leq \delta$ for each $i$, and so 
   \[
   \sum_{j \in V}  a_j Q(F_j, y')  \leq \sum_{j \in V} a_j  \leq \sum_{i \in A} a(H_i) \leq t \delta  \qedhere
   \]
\end{proof}

\begin{proposition}
  \label{sparse-qbnd}
  If $(x,y)$ is a $\delta$-sparse fractional LP solution with respect to weighting function $a$ and $\tilde Y = \textsc{KPR}(\mathcal G, M, y, t)$ for some integer $t$ with $12 m \leq t < 2/\delta$, then we have
$$
\sum_j a_j \E[Q(F_j, \tilde Y)] \leq 1/e + O(m^2 \delta \log \tfrac{1}{\delta t})
$$
\end{proposition}
\begin{proof}
It is convenient to take the following alternative,  slowed-down view of the \textsc{KPR} rounding process, where $k \geq 0$ is a parameter to be determined.
% \begin{algorithm}[H]
 % \caption{Alternate view of $y' = \textsc{KPR}(G_i, M, y, t)$}
\begin{algorithmic}[1]
  \STATE set $y^0 = \textsc{KPR}(\mathcal G, M, y, t 2^k)$ 
  \STATE \textbf{for} $\ell = 1, \dots, k$ \textbf{do} set $y^{\ell} = \textsc{KPR}(\mathcal G, M, y^{\ell-1}, 2^{k-\ell} t)$ 
    \STATE Output $\tilde Y = y^{k}$
\end{algorithmic} 
%\label{algo:knapsack_srdrx3}
%\end{algorithm}
  Let us define $C_\ell = \sum_{j \in \S} a_j Q(F_j, y^{\ell})$ for $\ell = 0, \dots, k$.  Thus we need to estimate $C_k = \sum_{j} a_j \E[Q(F_j, \tilde Y)]$. To do so, we will compute $\E[C_0]$ and $\E[C_{\ell} - C_{\ell-1}]$ for all $\ell = 1, \dots, k$.
 
For the $C_0$ term, we use Theorem~\ref{final-thm2}(b) to get
 \begin{equation}
 \label{hheqn16}
 \E[C_0] = \sum_{j \in \C} a_j \E[Q(F_j, y^{0})] \leq \sum_{j \in \C} a_j ( Q(F_j, y) + O(m^2/(t 2^k)))
\end{equation}
 
By Proposition~\ref{bsl-guarantee-knap}, we have $Q(F_j, y) \leq 1/e$ and therefore $\bE[C_0] \leq 1/e + O(m^2/(t 2^k))$. 

Next, suppose we condition on vector $y^{\ell-1}$ for some $\ell > 0$. By Theorem~\ref{final-thm2}(b), every client $j \in \C$ has
$$
\bE[Q(F_j, y^{\ell}) \mid y^{\ell-1}]  \leq Q(F_j, y^{\ell-1}) + O \bigl( \frac{m^2}{t 2^{k-\ell}} \bigr).
$$
Define the set of clients $A = \bigcup_{i: y^{\ell - 1} \in (0,1)} H_i$. By Property (E5), the vector $y^{\ell-1}$ has at most $2 t  \cdot 2^{k-\ell}$ fractional entries $i$, and by $\delta$-sparsity each such facility $i$ has $a(H_i) \leq \delta$. Thus, $\sum_{j \in A} a_j \leq (2 t \cdot 2^{k-\ell}) \delta$. Now observe that for $j \notin A$ we have $Q(F_j, y^{\ell}) = Q(F_j, y^{\ell-1})$ with probability one (both are equal to zero  or one). From this and Eq.~(\ref{hheqn16}), we get
$$
\bE[C_{\ell} - C_{\ell-1} \mid y^{\ell-1}] = \sum_{j \in A} a_j (\E[Q(F_j, y^{\ell}) \mid y^{\ell-1}] - Q(F_j, y^{\ell-1})) \leq \sum_{j \in A} a_j \cdot O \Bigl( \frac{m^2}{t 2^{k - \ell}} \Bigr)  = O( \delta m^2 )
$$

This implies that $\E[C_{\ell} - C_{\ell-1}] \leq O(\delta m^2)$, and so summing over $\ell$ gives
$$
\E[C_k] = \E[C_0] + \sum_{\ell=1}^k \E[ C_{\ell} - C_{\ell-1}] \leq 1/e + O \Bigl( \frac{m^2}{t 2^k} \Bigr) + O( \delta k m^2)
$$

Setting $k = \lceil \log_2 \frac{1}{\delta t} \rceil$ gives $\E[C_k] = 1/e + O( m^2 \delta \log \frac{1}{\delta t})$. (Note $k \geq 0$ by our assumption on $t$.)
\end{proof}

\subsection{Proof of Theorem \ref{res:standardKC}}
\label{sec:knapsack-center-m1}
When $m = 1$ (the single-knapsack center problem), we use the following Algorithm~\ref{algo:knapsack_srdr} to satisfy the knapsack constraint with no violation while guaranteeing that all clients get an expected approximation ratio arbitrarily close to $1+2/e$.
\begin{algorithm}[h]
\caption{$\textsc{SingleKnapsackCenterRound}\left(M, a, \delta \right)$ }
\begin{algorithmic}[1]
  \STATE Use Proposition~\ref{prune-prop} to obtain a $\delta$-sparse fractional solution $(x,y)$ with respect to $a$.
  \STATE Let $\tilde Y = \textsc{KPR}(\mathcal G, M, y, 12)$ 
\FOR{each block $G$}
\STATE open the facility $i \in G$ with $y_i > 0$ which has the \emph{smallest} weight $M(i)$.
\ENDFOR
\end{algorithmic} 
\label{algo:knapsack_srdr}
\end{algorithm}

{
\renewcommand{\thetheorem}{\ref{res:standardKC}}
\begin{theorem}
  For any $\gamma \in (0,1)$, a $\gamma$-fair feasible solution can be obtained in $n^{\tilde O(1/\gamma)}$ runtime.
\end{theorem}
\addtocounter{theorem}{-1}
}
\begin{proof}
Without loss of generality assume $\delta$ is sufficiently small.   In light of Lemma~\ref{glem1}, and by rescaling $\gamma$, it suffices to show that for any weighting function $a$, Algorithm~\ref{algo:knapsack_srdr} generates a solution set $\S$ satisfying $\sum_j a_j d(j, \S) \leq 1 + 2/e + O(\gamma)$ and $\max_{j} d(j, \S) \leq 3$.

It is useful to view lines 3--4 of Algorithm~\ref{algo:knapsack_srdr} as a two-part process. First, we convert the fractional vector $\tilde Y$ into an integral vector $z \in \{0, 1 \}^n$, by moving all the mass in each block to the item with smallest weight. We then open all facilities $i$ with  $z_i = 1$. Note that at most $24$ entries of $\tilde Y$ are modified compared to $z$.
  
First, we claim that $\S$ is feasible. For,  the fractional solution $(x, y)$ satisfies $M y \leq 1$, and Proposition~\ref{proposition:depround2} ensures that the vector $\tilde Y$ satisfies $M \tilde Y = M y$. The modification process can only decrease $M z$, so $M(\S) = M z \leq M \tilde Y$. Finally, by Proposition~\ref{bsl-guarantee-knap}, every client $j$ has $d(j, \S) \leq 3$ with probability one.
  
  We now turn to analyzing the connection cost.  By Proposition~\ref{bsl-guarantee-knap} we have:
  \begin{align*}
  \sum_{j \in \C} a_j d(j, \S) \leq \sum_{j \in \C} a_j (1 + 2 Q(F_j, z))
   \end{align*}
   
   Since at most $24$ entries of $y$ are modified to get the integral vector $z$, Proposition~\ref{mod-prop} gives
   $$
  \sum_{j \in \C} a_j d(j, \S) \leq O(\delta) + \sum_{j \in \C} a_j (1 + 2 Q(F_j, y))
   $$
    
By Proposition~\ref{sparse-qbnd} (noting that $m = 1$ and $t = 12$), we can take expectations of this quantity to get
 $$
  \sum_{j \in \C} a_j \bE[d(j, \S)] \leq O(\delta) + O(\delta \log \tfrac{1}{\delta}) + \sum_{j \in \C} a_j (1 + 2/e)
  $$

Setting $\delta = \gamma/\log(1/\gamma)$, we get $\E[ \sum_j a_j d(j, \S) ] \leq 1 + 2/e + O(\gamma).$  By running for an expected $O(1/\gamma)$ iterations of this process, we obtain a solution $\S$ with $\sum_j a_j d(j, \S) \leq 1 + 2/e + O(\gamma)$.  The overall runtime is $O(1/\gamma) n^{O(1/\delta)} = n^{O(\log(1/\gamma)/\gamma)}$.
\end{proof}

\subsection{Proof of Theorem \ref{res:MKC}}
\label{sec:knapsack-center-depround}
When there are multiple knapsack constraints, then the final rounding step cannot satisfy them all exactly. Instead, 
we use independent selection to obtain an additive pseudo-approximation, as shown in Algorithm~\ref{algo:kdr2}.  

 \begin{algorithm}[H]
\caption{$\textsc{MultiKnapsackCenterRound}\left(M, a, \delta, t \right)$ }
\begin{algorithmic}[1]
\STATE Use Proposition~\ref{prune-prop} to obtain a $\delta$-sparse fractional solution $x,y$ with respect to weighting function $a$.
\STATE Let $Y = \textsc{FullKPR}(G_{\ell}, M, y,t)$.
\RETURN $\S = \{ i  \mid Y_i = 1 \}$
\end{algorithmic} 
\label{algo:kdr2}
\end{algorithm}

Depending on the parameter $\delta$, there is a trade-off between approximation ratio, budget violation, and running time. We summarize this as follows:
 {
\renewcommand{\thetheorem}{\ref{res:MKC}}
\begin{theorem}
  Consider a multi-knapsack center instance with with $m \geq 1$ constraints and let $\gamma, \epsilon \in (0,1)$.  
  
  \begin{enumerate}
  \item[(a)] A $\gamma$-fair $O \bigl( \tfrac{m \sqrt{ \log(m/\gamma)}}{\sqrt{\gamma}} \bigr)$-additive pseudo-solution can be obtained in $\poly(n/\gamma)$ runtime.
      
\item[(b)] A $\gamma$-fair $O(\sqrt{m \log m})$-additive pseudo-solution can be obtained in $n^{\tilde O(m^2/\gamma)}$ runtime.

  \item[(c)] A $\gamma$-fair $\epsilon$-multiplicative pseudo-solution can be obtained in $n^{\tilde O(m^{3/2}/\epsilon + m^2/\gamma)}$ runtime.
  \end{enumerate}
\end{theorem}
\addtocounter{theorem}{-1}
}
\begin{proof}
  In light of Lemma~\ref{glem1}, and by rescaling $\gamma$, it suffices to show that for any weighting function $a$, we can get a solution $\S$ of the above form satisfying $\sum_j a_j d(j, \S) \leq 1 + 2/e + O(\gamma)$ and $\max_{j} d(j, \S) \leq 3$.
    
For result (a),  we use Algorithm~\ref{algo:kdr2} with $t = m^2/\gamma$ and $\delta = 1$. By Proposition~\ref{bsl-guarantee-knap} we have $\sum_j a_j d(j, \S) \leq \sum_j a_j (1 + 2 Q(W,Y))$. By Theorem~\ref{final-thm2}(b) we have $\E[Q(F_j,Y)] \leq Q(F_j, y) + O( m^2/t )$. By Proposition~\ref{bsl-guarantee-knap}, this in turn is at most $1/e + O(m^2/t)$. Since $m^2/t = \gamma$, we have
 $$
\E[ \sum_j a_j d(j, \S)]  \leq \sum_{j} a_j (1 + 2/e + O(\gamma)) = 1 + 2/e + O(\gamma)
$$

Let $\mathcal E$ be the desired event that $\S$ is a $q$-additive pseudo-solution for $q = O(\sqrt{t \log(m/\gamma)})$.  By Theorem~\ref{full-KPR-thm1}, we have $\Pr(\mathcal E) \geq 1 - \gamma$ and so $\bE[\sum_j a_j d(j, \S) \mid \mathcal E] \leq \frac{1 + 2/e + O(\gamma)}{ 1 - \gamma } \leq  1 + 2/e + O(\gamma)$. Thus, after an expected $O(1/\gamma)$ repetitions, we get a solution $\S$ which is a $q$-additive pseudo-solution and which has $\sum_j a_j d(j, \S) \leq 1 + 2/e + O(\gamma)$. With this choice of $\delta$, Algorithm~\ref{algo:kdr2} runs in $\poly(n/\gamma)$ time.

For result (b), we use Algorithm~\ref{algo:kdr2} with parameters $t = 12 m^2$ and $\delta = \frac{\gamma}{m^2 \log(1/\gamma)}$.  We can break the process of generating $Y$ into two steps:  we first generate  $\tilde Y = \textsc{KPR}(\mathcal G, M, y, 12 m^2)$ and then generate $Y = \textsc{IndSelect}(\mathcal G, \tilde Y)$. We may assume $\gamma$ is smaller than any needed constant, so $t < 2/\delta$. Therefore, Proposition~\ref{sparse-qbnd} gives
$$
\E \bigl[ \sum_j a_j (1 + 2 Q(F_j, \tilde Y)) \bigr] \leq 1 + 2/e + O(m^2 \delta \log \tfrac{1}{\delta t}).
$$

Our choice of $\delta$ and $t$ ensures this is at most $1 + 2/e + O(\gamma)$. So after an expected $\Omega(1/\gamma)$ repetitions the solution $\tilde Y$ satisfies $\sum_j a_j (1 + 2 Q(F_j, \tilde Y)) \leq 1 + 2/e + O(\gamma)$. Now suppose that this event has occurred, and let us condition on the fixed vector $\tilde Y$. Proposition~\ref{bsl-guarantee-knap} shows that $\sum_j a_j d(j, \S) \leq \sum_j a_j (1 + 2 Q(F_j, Y))$. The vector $Y$ is derived by modifying at most $2 t$ fractional entries of vector $\tilde Y$. Therefore, by Proposition~\ref{mod-prop}, we have
$$
\sum_j a_j (1 + 2 Q(F_j, Y)) \leq 2 \delta t + \sum_j a_j (1 + 2 Q(F_j, \tilde Y)) \leq 1 + 2/e + O(\gamma) 
$$

Furthermore, by Theorem~\ref{full-KPR-thm1}, the resulting solution $Y$ is an $O(\sqrt{m \log m})$-additive pseudo-solution  with probability $\Omega(1)$. (We are using the fact that this result holds even after conditioning on $\tilde Y$). Integrating over $\tilde Y$, we see that the vector $Y$ has the desired properties with probability $\Omega(1/\gamma)$. Since Algorithm~\ref{algo:kdr2} takes $n^{O(1/\delta)}$ time,  the overall expected runtime  is $O(1/\gamma) n^{O(1/\delta)} = n^{O(m^2 \log(1/\gamma)/\gamma)}$.

For result (c),  say facility $i$ is \emph{big} if $M_{k}(i) \geq \rho = \Theta( \frac{\epsilon}{\sqrt{m \log m}})$ for any constraint $k$. We can guess the big facilities in an optimal solution in $n^{O(m/\rho)}$ time. We remove all other big facilities, and apply result (b) to the residual instance. This gives a solution $\S$ with $M(\S)  \leq (1 + O( \sqrt{m \log m}) \cdot \rho) \vec 1 \leq (1+\epsilon) \vec 1$. The runtime is $n^{O(m/\rho + m^2 \log(1/\gamma)/\gamma)} = n^{O(m^{3/2} \log m/\epsilon + m^2 \log(1/\gamma)/\gamma)}$.
\end{proof}

By contrast, independent rounding would require $n^{O(m \log(m/\gamma)/\epsilon^2)}$ time for an $\epsilon$-multiplicative pseudo-solution, which is a significantly worse dependence upon $\epsilon$.  

\section{Further correlation bounds for KPR}
  \label{advanced-cor-sec}
In this section, we show some additional near-negative-correlation properties for KPR. Although these are not directly used by our clustering algorithms, they may be useful elsewhere.  We also remark that the dependent rounding algorithm of \cite{Bansal18} has been specifically designed to give concentration bounds, which would be similar to (and incomparable in strength with) the ones we develop here.
  
  \subsection{Analyzing property (E1) for small $Q(W,y)$}
  \label{e1-small-sec}
The additive gap in Theorem~\ref{final-thm2}(b) can make it unsuitable when $Q(W,y)$ is small. Theorem~\ref{final-thm2}(c) has a multiplicative gap, but has an undesirable dependence on the size of $W$.  Although we cannot achieve a multiplicative gap independent of the size of $W$, we can get something which is somewhat in between an additive and multiplicative gap. We show the following main result:
  
\begin{theorem}
  \label{final-thm}
Let $W \subseteq U$ and $t \geq 12 m$, and define $\theta = m^2/t$. Let $\tilde Y = \textsc{KPR}(\mathcal G, M, y,t)$. Then, for any $b > 2$ there holds $\E[Q (W,\tilde Y)] \leq e^{O(\theta^2 \log^3 b)} (Q (W,y) + \theta/b)$.
\end{theorem}

Let us fix $W, t, y$ for the remainder of this section and define $\theta = m^2/t$. Our overall strategy will be to solve a recurrence relation on $Q(W,y)$. One significant complication, which requires much technical delicacy, is that the intermediate values of $Q(W,y)$ are random variables. 

Consider the main loop of \textsc{KPR}, and define $y^i$ to be the state vector after the $i^{\text{th}}$ iteration of applying \textsc{KPR-iteration}. Let us select an integer parameter $k > 20$, whose role will be clarified later. For each $u = 1, \dots, k$ we define $I_u$ to be the first iteration $i$ such that  $T(y^i) \leq t$ or $Q(W, y^i) > \alpha_u$, where we define $\alpha_u = 2^{u-k}$. We also define $\alpha_0 = 0$ and $I_0 = 0$. Thus, $y^0$ is the vector after \textsc{IntraBlockReduce}. Also $y^k$ is the final output vector $\tilde Y$  since $\alpha_k = 1$.

 Throughout we write $S_i = Q(W,y^i)$ and $T_i = T(y^i)$. For each value $u = 0, \dots, k$ we define $H_u = S_{I_u}$, and so $H_0 = Q(W, y^0)$ while $H_k$ is the final value $Q(W, \tilde Y)$.

There are three stages to analyze the evolution of $\bE[S_i]$. First, we analyze the change in a single round (going from $S_{i}$ to $S_{i+1}$). Second, we analyze the change over each value of $u$ (going from $H_{u}$ to $H_{u+1}$). Finally, we analyze the total change from $H_{0}$ to $H_{k}$.

\begin{proposition}
\label{prop3}
Let $u \geq 0$, and suppose that we condition on all state up to round $i$ with $i < I_{u}$. If $t \geq 100 m k$, then we have $\bE[ S_{i+1} ] \leq S_{i} + O(  \alpha_u  m^2 k^2 /  T_i^2 )$.
\end{proposition}
\begin{proof}
By Lemma~\ref{prop2a}, we have:
\begin{align*}
\E[S_{i+1}] &\leq S_{i} \cosh\Bigl(\frac{6 m}{T_i} \sum_G y^i(G \cap W) \Bigr) \leq S_{i} \cosh \Bigl(\frac{6 m}{T_i} \sum_G \ln(1-y^i(G \cap W) \Bigr)  = S_{i} \cosh\Bigl(\frac{6 m \ln S_{i}}{T_i} \Bigr)
\end{align*}
where the second inequality uses the fact that $x \leq -\ln(1-x)$ for $x \in [0,1]$.  Also,  by definition of $I_u$, we have $S_i \leq \alpha_u$ and $T_i \geq t$.   To finish, we claim that
\begin{equation}
  \label{si-prop}
S_{i} \cosh\left(\frac{6 m \ln S_{i}}{T_i} \right) - S_{i} \leq O \left( \frac{\alpha_u m^2  k^2 }{T_i^2} \right).
\end{equation}

To show this, let $\beta = 6 m / T_i$, and observe that since $T_i \geq t \geq 100 mk$ we have $\beta \leq 1/k$. Now consider the function $f(s) =  s \cosh(\beta \ln s) - s$. Simple analysis shows that $f(s)$ is an increasing function for $s \leq w = ( \frac{1-\beta}{1+\beta} )^{1/\beta}$. Furthermore, the restriction that $\beta \leq 1/k \leq 1/20$ ensures that $w = \Theta(1)$. So if $\alpha_u \leq w$, then $f(S_i) \leq f(\alpha_u) = \alpha_u \cosh(\beta \ln \alpha_u) - \alpha_u$, and a second-order Taylor series for $\cosh$ then shows that $f(\alpha_u) \leq O(\alpha_u \beta^2 \log^2 \alpha_u) \leq O(\alpha_u \beta^2 k^2)$, implying Eq.~(\ref{si-prop}). If $\alpha_u \geq w$, then Proposition~\ref{tech-prop3} shows that $f(S_i) \leq O(\beta^2) = O(\alpha_u \beta^2)$ and again Eq.~(\ref{si-prop}) holds.
\end{proof}

\begin{proposition}
\label{prop3a}
If $t \geq 100 m k$, then $\bE[ H_{u+1} \mid H_u ] \leq H_{u} + [[ H_{u} \geq \alpha_{u} ]] \cdot O( \alpha_{u+1} \theta  k^2)$ for each $u$.
\end{proposition}
\begin{proof}
Suppose we condition on $I_u$ as well as all state up to iteration $i = I_u$, including the random variable $H_u$. By definition of $I_u$, we must have either $S_i \geq \alpha_u$ or $T_i < t$. In the latter case,  we immediately have $I_{u+1} = i$ as well and so $H_{u+1} = H_u$.   So, let us assume that $S_i \geq \alpha_{u}$ and we want to show that $\bE[ H_{u+1} ] \leq H_{u} + O( \alpha_{u+1}  \theta k^2 )$. 

For each $j \geq i$, define $\tilde S_j = S_{\min(j, I_{u+1})}$. Note that $\tilde S_i = S_i = H_u$ and $\lim_{j \rightarrow \infty} \tilde S_j = H_{u+1}$. We claim that, for all $v \geq i$, there holds
  \begin{equation}
  \label{tildes-eqn}
  \bE[\tilde S_{v+1} \mid \tilde S_v] \leq \tilde S_v + [[T_v > t]] \cdot  O \bigl( \alpha_{u+1} m^2 k^2 / T_v^2 \bigr)
  \end{equation}

  For, if $v \geq I_{u+1}$ or $T_v \leq t$, we have $\tilde S_{v+1} = \tilde S_v$ with probability one; otherwise, we have $S_v = \tilde S_v$ and so this follows from Proposition~\ref{prop3}. Since $\tilde S_i = H_u$, we can sum Eq.~(\ref{tildes-eqn}) over $v = i, \dots, j-1$ and use iterated expectations to obtain:
$$
\E[\tilde S_{j}] \leq \tilde S_i + O(m^2 \alpha_u k^2) \cdot \sum_{v=i}^{j-1} \frac{ \Pr(T_v > t) }{ T_v^2} = H_u + O(m^2 \alpha_{u+1} k^2)  \sum_{\ell=t+1}^{\infty} \sum_{v=i}^j \frac{ \Pr( T_v = \ell ) }{\ell^2}
$$

By Proposition~\ref{prog-prop}, the vector $y$ gains an integral entry with probability at least $0.24$ in each round. This implies that the number of iterations $v$ with $T_{v} = \ell$, is stochastically dominated by a $\text{Geometric}(0.24)$ random variable. Thus $\sum_{v=i}^{\infty} \Pr(T_{v} = \ell) \leq 1/0.24 \leq 5$ and so 
$$
\E[\tilde S_j] \leq H_u + O( m^2 \alpha_{u+1} k^2 ) \cdot \sum_{\ell=t+1}^{\infty} 5/\ell^2 \leq H_u + O( m^2 \alpha_{u+1} k^2/t ) 
$$

Since $H_{u+1} = \lim_{j \rightarrow \infty} \tilde S_j$ and $m^2/t = \theta$, this implies that $\E[H_{u+1}] \leq H_u + O( \alpha_{u+1} \theta  k^2)$.
\end{proof}

\begin{proposition}
\label{prop4}
We have $\bE[ H_{k} ] \leq \bigl( \bE[H_{0}] + O(\theta k^2 2^{-k}) \bigr) e^{O( \theta k^3)}$.
\end{proposition}
\begin{proof}
If $t < 100 m k$, then $( \bE[H_0] + \theta k^2 2^{-k}) e^{\theta k^3} \geq e^{-\Omega(k)}  \cdot e^{0.01 m k^2} \geq \Omega(1)$. Since $H_{k} \leq 1$, the claimed bound will then hold vacuously. So assume that $t \geq 100 m k$. By Proposition~\ref{prop3a} and iterated expectations, we have for each $u \geq 0$:
\begin{equation}
  \label{ppe1}
\E[ H_{u+1} ]  \leq \bE[H_{u}] + \Pr( H_{u} \geq \alpha_{u} ) \cdot O(  \alpha_{u+1} \theta k^2 )
\end{equation}

For $u = 0$, recall that $\alpha_{0} = 0, \alpha_1 = 2^{1-k}$; thus Eq.~(\ref{ppe1})  implies $\bE[H_{1}] \leq \bE[H_{0}] + O(\frac{\theta k^2}{2^{k}})$. For each $u \geq 1$, Markov's inequality applied to Eq.~(\ref{ppe1}) gives
\begin{align*}
\bE[H_{u+1}] &\leq \bE[H_{u}] + O \Bigl( \frac{\bE[H_{u}]}{\alpha_{u}} \cdot \alpha_{u+1}  \theta k^2 \Bigr) = \bE[H_{u}] (1 + O(  \theta k^2))
\end{align*}

Combining these bounds for $u = 0, \dots, k-1$ gives
\[
\bE[H_{k}] \leq \bigl( \bE[H_{0}] + O( \tfrac{\theta k^2}{2^k}) \bigr) \bigl( 1 + O(\theta k^2) \bigr)^{k} \leq \bigl( \bE[H_{0}] + O(  \tfrac{\theta k^2}{2^k}) \bigr) e^{O(\theta k^3)} \qedhere
\]
\end{proof}

\begin{proof}[Proof of Theorem~\ref{final-thm}]
Set $k = \lceil c \log_2 b \rceil$, for some constant $c >  20$.  By Proposition~\ref{prop4}, we have
$$
  \E[H_{k}] \leq \bigl( \bE[H_{0}] + \frac{K \theta c^2 \log^2 b }{b^c}  \bigr) e^{K c^3 \theta \log^3 b} 
  $$
  for an absolute constant $K \geq 1$. Choosing $c$ to be a sufficiently large constant gives  $K c^2 \log^2 b/ (b^c) \leq 1/b$ for all $b > 2$, and so $\E[H_{k}] \leq \bigl( \bE[H_{0}] +  \theta/b  \bigr) e^{K c^3 \theta \log^3 b}$. Finally, recall that $Q(W,\tilde Y) = H_k$ and by Proposition~\ref{igprop} we have $\bE[H_{0}] = Q(W,y)$. 
\end{proof}

  \subsection{Concentration bounds for \textsc{FullKPR}}
The analysis in Section~\ref{e1-small-sec} can be used show a lower-tail concentration bound for \textsc{FullKPR}.
  \begin{theorem}
 Let $Y = \textsc{FullKPR}(\mathcal G, M, y,t)$ for $t > 12 m$ and let $\theta = m^2/t$. If $y \bullet w \geq \mu$ for some parameters $\mu \geq 1$ and $w \in [0,1]^U$, then
  $$
  \Pr(Y \bullet w \leq \mu (1 - \delta)  ) \leq e^{O(\theta (1 + \delta \mu)^3)} \cdot \Bigl( \frac{e^{-\delta}}{(1-\delta)^{1-\delta}} \Bigr)^{\mu}
  $$
\end{theorem}
\begin{proof}
  Let $\mathcal E$ be the event $Y \bullet w \leq \mu (1 - \delta)$. Consider forming a random set $W$, wherein each $j \in U$ goes into $W$ independently with probability $q w_j$ for some parameter $q \in [0,1]$ to be specified.  We will compute $\bE[Q(W, Y)]$ in two different ways.   First, suppose we condition on the event $\mathcal E$, as well as all the random variables $Y$. Then by Proposition~\ref{tech-prop2},
  \begin{align*}
    \bE[Q(W,Y) \mid \mathcal E, Y] &=  \bE \Bigl[ \prod_{G \in \mathcal G} (1 - q \sum_{j \in G} w_j Y_j) \mid \mathcal E, Y \Bigr] \geq \bE[ (1-q)^{Y \bullet w} \mid \mathcal E, Y] \geq (1-q)^{\mu (1-\delta)}
\end{align*}

  Therefore, $\bE[Q(W,Y)] \geq \Pr(\mathcal E) (1-q)^{\mu (1-\delta)}$.  On the other hand, if we condition on the random variable $W$ then Theorem~\ref{final-thm} gives $\E[Q (W,Y) \mid W] \leq e^{O(\theta \log^3 b)} ( Q (W,y) + \theta/b)$ for any parameter $b > 2$. By the way we form $W$, we can compute 
  $$
  \bE[Q(W,y)] = \bE \Bigl[ \prod_{G \in \mathcal G} (1 - \sum_{j \in G} y_j [[j \in W]]) \Bigr] =  \prod_{G \in \mathcal G} (1 - q \sum_{j \in G} w_j y_j) \leq e^{-\sum_j q w_j y_j} = e^{-q \mu}
  $$
  and hence we have 
  $$
    \E[Q (W,Y) ] \leq e^{O(\theta \log^3 b)} ( e^{-q \mu} + \theta/b)
    $$

   Putting these two bounds together, we see $\Pr(\mathcal E) \leq \frac{e^{O(\theta \log^3 b)} (e^{-q \mu} + \theta/b)}{(1-q)^{\mu(1-\delta)}}$.  Let us now set $q = \delta$ and $b = e^{1 + \delta \mu}  > 2$. With these parameters, we have
\[
\Pr(\mathcal E) \leq \frac{e^{O(\theta (1+\delta \mu)^3)} ( e^{-\delta \mu} +\theta e^{-1-\delta \mu})}{(1 - \delta)^{\mu (1 - \delta)}}  \leq e^{O(\theta (1 + \delta \mu)^3)} \Bigl( \frac{e^{-\delta}}{(1-\delta)^{1-\delta}} \Bigr)^{\mu} \qedhere
\]
\end{proof}

To show upper tail bounds, we use an approach of \cite{schmidt1995chernoff} based on symmetric polynomials. We also use a number of extremal bounds from that paper for such polynomials.

\begin{theorem}
Let  $Y = \textsc{FullKPR}(\mathcal G, M, y,t)$ and let $\theta = m^2/t$. If $y \bullet w \leq \mu$ and $t > 10000 m (1+\delta \mu)$ for some parameters $\mu \geq 1$ and $w \in [0,1]^U$, then
$$
  \Pr \bigl( Y \bullet w \geq \mu(1+\delta)  \bigr) \leq  e^{O(\theta (1+ \delta \mu)^2)}  \cdot \Bigl( \frac{e^{\delta}}{(1+\delta)^{1+\delta}} \Bigr)^{\mu}
  $$
\end{theorem}
\begin{proof}
Let us consider the random variable defined by
$$
H = \sum_{\substack{L \subseteq U \\ |L| = k}} \prod_{j \in L} w_j Y_j
$$
where $k = \lceil \mu \delta \rceil$. We use Proposition~\ref{lprop1} to compute the expectation of $H$:
\begin{align*}
\bE[H] &= \sum_{\substack{L \subseteq U \\ |L| = k}} (\prod_{j \in L} w_j) \bE \bigl[ \prod_{j \in L} Y_j  \bigr] \leq e^{O(\theta k^2 k^2)} \sum_{\substack{L \subseteq U \\ |L| = k}} \prod_{j \in L} w_j y_j~.
\end{align*}
This is precisely the $k^{\text{th}}$ symmetric polynomial applied to the quantities $w_j y_j$. As shown in \cite{schmidt1995chernoff}, there holds:$$
\sum_{\substack{L \subseteq U \\ |L| = k}} \prod_{j \in L} w_j y_j \leq \frac{ (\sum_j w_j y_j )^k }{k!} = \frac{\mu^k}{k!}
$$

As shown in \cite{schmidt1995chernoff}, whenever $k \leq \lfloor a \rfloor$ and $Y \bullet w \geq a$ for a real number $a \geq 0$, there holds $H \geq \binom{a}{k}$. We use this fact with $a = \mu(1+\delta)$; note $k \leq \lfloor \mu(1+\delta) \rfloor$ since $\mu \geq 1$.   Applying Markov's inequality to $H$ gives
$$
\Pr \bigl( Y \bullet w \geq \mu(1+\delta)  \bigr) \leq \frac{\bE[H]}{ \binom{\mu(1+\delta)}{k} } \leq e^{O(\theta k^2)} \frac{  \mu^k/k! }{ \binom{\mu(1+\delta)}{k} } 
  $$
  
  Finally, as shown in \cite{schmidt1995chernoff}, the value $k = \lceil \mu \delta \rceil$ ensures that  $ \frac{  \mu^k/k! }{ \binom{\mu(1+\delta)}{k} } \leq  \Bigl( \frac{e^{\delta}}{(1+\delta)^{1+\delta}} \Bigr)^{\mu}$. So for $t > 10000 m (1+\delta \mu)$, we have shown that
  \[
  \Pr \bigl( Y \bullet w \geq \mu(1+\delta)  \bigr) \leq  e^{O(\theta (1+\delta \mu)^2)}  \cdot \Bigl( \frac{e^{\delta}}{(1+\delta)^{1+\delta}} \Bigr)^{\mu} \qedhere
  \]
\end{proof}

\section{Acknowledgments}
We thank Nikhil Bansal, Chandra Chekuri, Shi Li, and the referees of the conference and journal versions of this paper for their helpful suggestions. 

\appendix

\section{Some technical lemmas}
\begin{proposition}
  \label{tech-prop2}
  For $Y \in \{0,1 \}^n, w \in [0,1]^n, \lambda \in [0,1]$, we have $ (1 - \lambda)^{Y \bullet w} \leq \prod_{G \in \mathcal G} (1 - \lambda \sum_{j \in G} w_j Y_j)$.
\end{proposition}
\begin{proof}
Since $Y$ is an integral vector and $Y(G) \leq 1$, we have $1 - \lambda \sum_{j \in G} w_j Y_j = \prod_{j \in G} (1 - \lambda w_j Y_j)$ for any block $G$.   Therefore, we get
\begin{align*}
\prod_{G \in \mathcal G} (1 - \lambda \sum_{j \in G} w_j Y_j) &= \prod_{G \in \mathcal G} (1 - \lambda \sum_{j \in G} w_j Y_j) = \prod_{G \in \mathcal G} \prod_{j \in G} (1 - \lambda w_j Y_j) = \prod_{i=1}^n (1 - \lambda w_j Y_j) \\
  & \geq \prod_{i=1}^n (1 - \lambda)^{w_j Y_j} \qquad \text{ as $(1 + a b) \geq (1 + a)^b$ for $a  \geq -1$ and $b \in [0,1]$} \\
    &= (1 - \lambda)^{Y \bullet w} \qedhere
\end{align*}
\end{proof}

\begin{proposition}
  \label{tech-prop3}
For $u \in \mathbb Z_{\geq 0}$ and  $a \in [0,\frac{1}{u+1}]$ and $s \in [0,1]$, we have  $s \cosh( a (u - \ln s) ) - s \leq a^2 (u+1)^2$.
\end{proposition}
\begin{proof}
  Let $f(s) = s \cosh( a (u - \ln s) ) - s$. The critical points of
  function $f(s)$ occur at $s_0 = e^{u}$ and $s_1 = e^u (\frac{1-a}{1+a})^{1/a}$. Since $s_0$
  is outside the allowed parameter range, this means that maximum value of $f(s)$ in the interval $[0,1]$
  must occur at either $s = 0, s = s_1$, or $s = 1$. 

  At $s = 0$, we have $f(s) = 0$. At $s = 1$, we have  $f(1) = \cosh(a u) -1$, which is at most $(a u)^2$ since $a u \leq 1$.  Let us now bound $s_1$. One can check that $(\frac{1-a}{1+a})^{1/a}$ is a decreasing function of $a$, Thus, as $a \leq 1/(u+1)$, we have $s_1 \geq e^u \bigl( \frac{1-1/(u+1)}{1+1/(u+1)} \bigr)^{u+1} = e^u \bigl( u/(u+2) \bigr)^{u+1}$.
This is larger than $1$ for $u \geq 3$, so it is out of the range of interest.

 So we only need to check $f(s_1) \leq a^2 (u+1)^2$ for $u = 0, 1, 2$; these are are all routine calculations.
\end{proof}

\begin{proposition}
\label{tech-prop4}
For real numbers $T, t, a$ with $T \geq t \geq 2 \sqrt{a} > 0$, we have
$$
e^{a(1/t - 1/(T-1))} \leq e^{a(1/t - 1/T)} - \frac{a e^{a(1/t - 1/T)}}{2 T^2}
$$
\end{proposition}
\begin{proof}
Dividing both sides by $e^{a(1/t - 1/T)}$, we need to show that $e^{- a/(T-1) + a/T} \leq (1 -  a/(2 T^2))$. Since $T \geq 2 \sqrt{a}$, we have $a/(2 T^2) \leq 1/8$. Therefore, $1 - a/(2 T^2) \geq e^{-3/4 (a/T^2)}$. So it suffices to show that $-a/(T-1) + a/T \leq -(3/4) a/T^2$. It is routine to verify this holds for $a > 0$ and $T \geq 2$.
\end{proof}

\bibliographystyle{abbrv}
\bibliography{ref}

\end{document}